\documentclass[a4paper,onecolumn,10pt,accepted=2025-03-08]{quantumarticle}
\pdfoutput=1
\usepackage{amsmath,mathtools,amsthm,amssymb}

\usepackage{tabularx}
\usepackage{units}
\usepackage{complexity}
\usepackage{graphicx}
\usepackage{color}
\usepackage{xcolor}
\usepackage{comment}
\usepackage{bbold}
\usepackage{lipsum} 
\usepackage{titlesec}
\usepackage{enumitem}
\usepackage[left=23mm,right=23mm,top=18mm,bottom=25mm,columnsep=16pt]{geometry} 
\usepackage{pifont}
\usepackage{comment}
\usepackage[linesnumbered,ruled,vlined]{algorithm2e}
\usepackage{listings}

\usepackage[colorlinks=true,linkcolor=teal,citecolor=teal,urlcolor=teal]{hyperref}

 \newcommand{\be}{\begin{equation}}
\newcommand{\ee}{\end{equation}}
\newcommand{\ba}{\begin{eqnarray}}

\newcommand{\ea}{\end{eqnarray}}
\usepackage{array}
\newcolumntype{H}{>{\setbox0=\hbox\bgroup}c<{\egroup}@{}}
\usepackage{graphics}
\usepackage[normalem]{ulem}
\usepackage{complexity}

\definecolor{myrefcolor}{rgb}{0.067,0.5,0.5}
\definecolor{myurlcolor}{rgb}{0.1,0,0.9}

\usepackage[toc,page]{appendix}

\usepackage{braket}
\usepackage{physics}
\usepackage{csquotes}

\definecolor{antonio}{rgb}{.2,.5,.1}
\newcommand{\antonio}
[1]{{\color{antonio}$\big[\![$ \raisebox{.7pt}{Antonio}\!\!\:\raisebox{-.7pt}{}: \textit{#1}$\ ]\!\big]$}}

\definecolor{len}{rgb}{.8,.5,.1}

\newcommand{\nx}{n_x}

\newtheorem{theorem}{Theorem}

\newtheorem{definition}{Definition}
\newtheorem{lemma}{Lemma}
\newtheorem{corollary}{Corollary}

\newtheorem{problem}{Problem}
\newcommand{\Pf}{\operatorname{Pf}}

\newcommand{\fu}{Dahlem Center for Complex Quantum Systems, Freie Universit\"{a}t Berlin, 14195 Berlin, Germany}
\allowdisplaybreaks
\begin{document}

\title{PAC-learning of free-fermionic states is \class{NP}-hard}

\author{Lennart Bittel}
\affiliation{\fu}

\author{Antonio Anna Mele}
\affiliation{\fu}

\author{Jens Eisert}
\affiliation{\fu}

\author{Lorenzo Leone}
\affiliation{\fu}
\begin{abstract}
Free-fermionic states, also known as Gaussian states or states prepared by matchgate circuits, constitute a fundamental class of quantum states due to their efficient classical simulability and their crucial role across various domains of physics. With the advent of quantum devices, experiments now yield data from quantum states, including estimates of expectation values. 
We establish that deciding whether a given dataset, formed by a few Majorana correlation functions estimates, can be consistent with a free-fermionic state is an \texttt{NP}-complete problem. Our result also extends to datasets formed by estimates of Pauli expectation values. This is in stark contrast to the case of stabilizer states, where the analogous problem can be efficiently solved. Moreover, our results directly imply that free-fermionic states are computationally hard to properly $\PAC$-learn, where PAC-learning of quantum states is a learning framework introduced by Aaronson. Remarkably, this is the first class of classically simulable quantum states shown to have this property. 
\end{abstract}

\maketitle

\section{Introduction}
Recent years have witnessed a significant advancement in the actual realization of quantum devices~\cite{Bluvstein_2023,Kim2023}, accompanied by a surge in the number of experiments conducted on them. Data obtained from these quantum experiments are now starting to be utilized as benchmarks for numerical simulations~\cite{Kim2023,begušić2023fast,liao2023simulation,rudolph2023classical,Begu_i__2024,patra2023efficient}; at the same time, these developments come along with challenging demands concerning the benchmarking and verification of those quantum devices \cite{Eisert_2020,PRXQuantum.2.010201}.

In this context, a natural question arises: can collected data be consistent with a particular class of quantum states? For instance, one might wonder whether some collected expectation values estimates could be consistent with a state within a classically tractable class, such as matrix product states~\cite{Cirac_2021}, stabilizer states~\cite{gottesman1998heisenberg}, or free-fermionic quantum states~\cite{Surace_2022}. This consideration might be driven by the desire to efficiently perform classical post-processing  with a state that matches such data and try to predict unmeasured quantities~\cite{Aaronson_2007,rocchetto_2019_PAC_experiment}. Alternatively, one might seek to validate, based on the obtained data, that a prepared quantum state does not belong to a classically tractable class of states~\cite{Oliviero_2022}. This motivates the problem of determining whether there exists a state within a fixed class which can be consistent with a given dataset composed of a few expectation value estimates.
Previous research has tackled this problem for stabilizer states and matrix product states~\cite{rocchetto2018stabiliser,yoganathan2019condition}, demonstrating that with respect to meaningful observable expectation values, such as Pauli expectation values, it is possible to efficiently decide whether a given dataset is consistent with the class of states or not. However, the analogous problem for the relevant class of free-fermionic quantum states remains open to date, and our work addresses this gap.

Free-fermionic quantum states~\cite{Surace_2022,BravyiFermions} are a fundamental class of states with applications across various domain of physics~\cite{Baxter:1982zz,Kitaev_2006,Echenique_2007,schrieffer2018theory}, especially in condensed matter physics and quantum chemistry. They are at the basis of the successful Hartree-Fock method and \emph{density-functional theory} (DFT)~\cite{RevModPhys.87.897,Martin_2004,giuliani2008quantum}, which have been at the forefront of computational calculations in solid-state physics and quantum chemistry. In quantum computing, free-fermionic quantum states are are recognized mainly because they can be simulated efficiently classically~\cite{Terhal_2002,Valiant,Jozsa_2008}. They are also referred to as fermionic Gaussian states~\cite{Surace_2022} or matchgates states~\cite{Terhal_2002,Valiant}.

\begin{figure}
\centering
\includegraphics[width=0.65\textwidth]{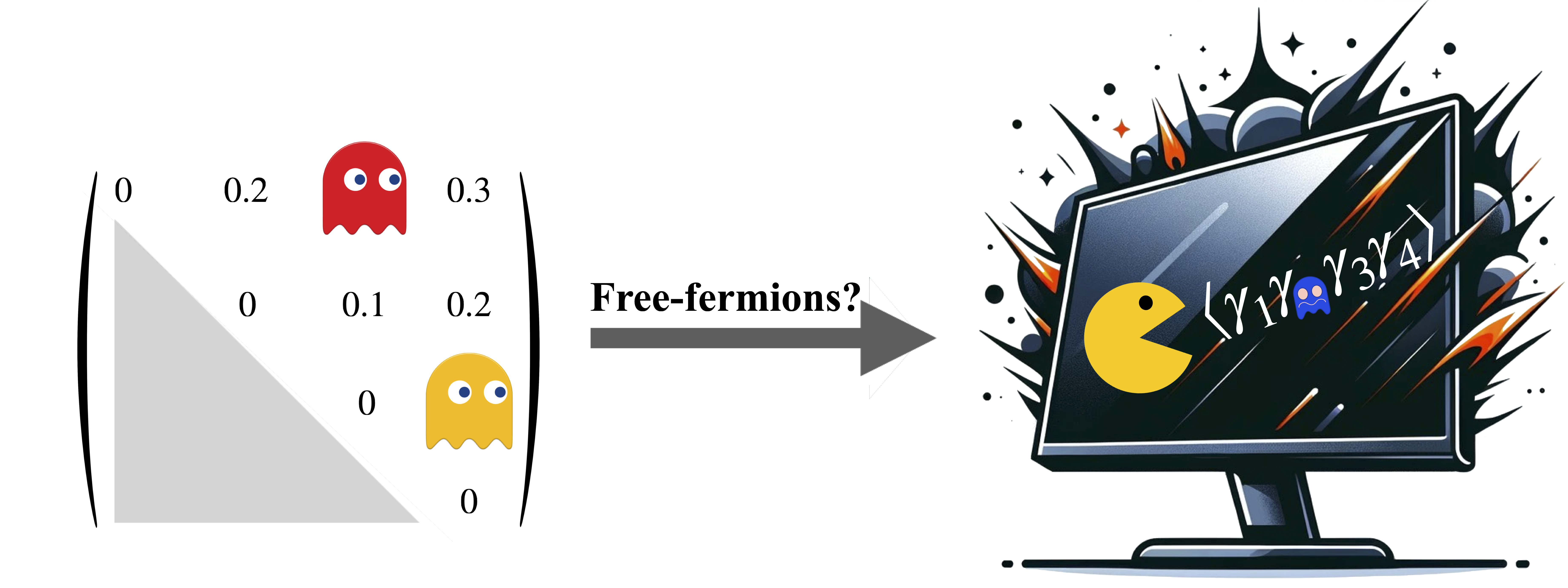}

\caption{Given partial information about the $2$-body \emph{correlation matrix} and a few $4$-body correlation functions $\langle \gamma_{i}\gamma_{j}\gamma_{k}\gamma_{l}\rangle$, determining whether there exists a free fermionic state consistent with the provided data is NP-complete. }
\label{Fig:circ}
\end{figure}

In this work, we show that deciding whether a given dataset, consisting of a few Majorana correlation functions estimates, can be consistent with a free-fermionic state is an \texttt{NP}-complete problem. 
Majorana correlation functions are significant and, presumably, the most natural expectation values to probe in the fermionic context~\cite{Naldesi_2023,denzler2023learning}, akin to Pauli expectation values in the qubit setting~\cite{Naldesi_2023,denzler2023learning}. Even when restricting datasets to include only few $2$- and $4$-body correlation function estimates, we show that this decision problem remains \texttt{NP}-complete. Leveraging the Jordan-Wigner mapping~\cite{Terhal_2002}, our findings extend to datasets containing estimates of Pauli expectation values. Our results also have implications for quantum learning theory, namely we solve the problem of understanding the computational complexity of `proper PAC-learning'~\cite{anshu2023survey,Aaronson_2007,rocchetto2018stabiliser} of free-fermionic quantum states, that we detail in the following subsection.

The intuition behind why the problem we study is \texttt{NP}-hard can be explained as follows. An $n$-mode free fermionic state is uniquely characterized by its correlation matrix, a $2n\times 2n$ matrix where each entry corresponds to a specific $2$-body correlation function. Thus, knowing the value of a $2$-body correlation function is equivalent to specifying an element of this matrix. However, higher-order correlation functions are known to impose nonlinear constraints on the elements of the correlation matrix (due to so-called Wick's theorem~\cite{Surace_2022}). This implies that deciding whether a given dataset, comprising of a few $2$-body and higher-order correlation functions, can be attributed to a free-fermionic state is equivalent to asking whether a $2n\times 2n$ matrix, with only a few of its elements fixed, can be completed so that it satisfies the prescribed nonlinear constraints. If higher-order correlation functions are not included in the constraints, one does not arrive at computationally difficult problems. 
In essence, our free-fermionic consistency problem can be viewed as a \emph{matrix completion} problem with nonlinear constraints, and it is well known that nonlinear optimization problems are prone to be \texttt{NP}-hard~\cite{arora_barak_2009}. We formally establish the \texttt{NP}-hardness of our problem through a polynomial reduction to the well-known and notoriously hard 3-SAT problem~\cite{arora_barak_2009}.
\subsubsection{Implication for PAC-learning of quantum states}
{The realization of increasingly complex quantum states in quantum devices has spurred significant interest in the problem of recovering properties of quantum states from collected measurements~\cite{anshu2023survey,Huang_2020,Kliesch_2021}—a task that can be analyzed through the lens of learning theory~\cite{anshu2023survey}.
Aaronson's seminal contribution~\cite{Aaronson_2007} introduced the paradigm of \emph{probably approximately correct} (PAC) learning quantum states, motivated by the practical observation that experimenters often cannot arbitrarily choose observables and must instead operate within constraints imposed by the available measurements.} In this framework, the primary objective is to predict unseen expectation values, based on a training dataset of few known expectation values. More precisely, this \textit{training set} consists of pairs of  \emph{positive operator-valued measure} (POVM) elements and their corresponding expectation values, drawn from a distribution over the POVM elements. The goal is to accurately predict the expectation values of unseen POVM elements that share the same underlying distribution.
Remarkably, Aaronson has demonstrated a surprising result: the entire class of $n$-qubit quantum states can be PAC-learned using training sets of size $O(n)$, where $n$ is the number of qubits. This result has subsequently been 
made use of in an experimental demonstration 
on a photonic device~\cite{rocchetto_2019_PAC_experiment}. 
The proposed algorithm~\cite{Aaronson_2007} involves generating a \emph{hypothesis} quantum state $\sigma$, which is consistent with all the expectation values in the training set within a given margin of error. However, while Aaronson's result achieves sample efficiency, it does not resolve the question of whether time-efficient PAC-learning of specific classes of quantum states is feasible, relating to notions of computational complexity. It is worth noting that the task of finding such a state $\sigma$ consistent with the given data can be cast as a \emph{semi-definite program}  (SDP)~\cite{aaronson2023efficient,rocchetto2018stabiliser} that, for general quantum states, can only be solved in $\exp(O(n))$ time. Time efficient PAC-learning can be achieved if the state $\sigma$ can be found in polynomial time and expectation values on unseen observables can be computed in polynomial time. For instance, it has been demonstrated that classes such as stabilizer states~\cite{rocchetto2018stabiliser} and states with low entanglement~\cite{yoganathan2019condition} can be time-efficiently PAC-learned, with respect to meaningful class of observables (associated with their respective POVMs).
However, the problem remains open for free-fermionic quantum states~\cite{BravyiFermions} -- constituting another noteworthy class of quantum states that can also be efficiently simulated classically~\cite{knill2001fermionic,Terhal_2002,Surace_2022}, akin to stabilizer states and low Schmidt rank states. Throughout our presentation, we will primarily refer to them as fermionic Gaussian states.

In this work, we show that ``proper"-PAC learning of fermionic Gaussian states is computationally difficult in the sense of being \texttt{NP}-hard, when considering POVM elements associated with Majorana observables. Here, ``proper" means that the hypothesis state $\sigma$ must belong to the same class (e.g., the class of fermionic Gaussian states) as the state from which the training set is originated. 

{
As discussed in Ref.~\cite{Liang_2023}, proper PAC-learning is a natural framework for PAC-learning a classically simulable class of quantum states, since it ensures time-efficient prediction of new (unseen) data while guaranteeing the existence of a valid hypothesis consistent with the training data.
However, requiring the hypothesis to belong to the same class (in this case, fermionic Gaussian states) may not be strictly necessary; it might suffice for the hypothesis to enable efficient prediction.
It is also important to emphasize that our definition of proper PAC-learning, for which we demonstrate hardness, pertains specifically to Aaronson’s information-theoretic algorithm~\cite{aaronson2023efficient}, which involves finding a matching state that fits the training data set. In particular, we show that finding a free-fermionic state consistent with the training data is computationally intractable. Nonetheless, there could, in principle, exist an alternative algorithm—unrelated to finding a matching state—that enables efficient generalization on a given training set drawn from a given probability distribution.
}

As mentioned in the previous section, we also address a substantially broader problem than proper PAC-learning, which we refer to as the $\class{Gaussian\,Consistency}$ problem, defined as the problem of deciding if there exists a fermionic Gaussian state that closely approximates a training set comprising $\Theta(n)$ Majorana or Pauli expectation values. 
It is important to note that, in the $\class{Gaussian\,Consistency}$ problem, the training data could potentially stem from a non-Gaussian state, unlike the proper PAC-learning scenario where the data must originate from a Gaussian underlying state. 
Through a straightforward reduction from the $\class{Gaussian\,Consistency}$ problem to the proper PAC-learning problem, we demonstrate the \texttt{NP}-hardness of the latter. 
\subsubsection{Different notions of learning}
{Beyond PAC-learning, which is the primary focus of this paper, various other learning frameworks have been explored in the literature, such as quantum state tomography~\cite{Haah_2017,odonnell2015quantum}, property testing~\cite{montanaro2018surveyquantumpropertytesting}, and shadow tomography~\cite{Aaronson_2007,Huang_2020}.}

Quantum state tomography involves determining, through measurements on copies of an unknown quantum state, a classical description of a quantum state that closely approximates the unknown target state within the trace distance metric~\cite{anshu2023survey}. Performing quantum state tomography without prior assumptions about the state to be learned requires an exponential number of copies and computational time scaling with the system size~\cite{anshu2023survey}. However, if one possesses prior assumptions, quantum state tomography may still be performed efficiently, as demonstrated for several classes of states~\cite{Cramer_2010,lanyonEfficientTomographyQuantum2017,fanizza2023learning,rouzé2023learning,huang2024learning,montanaro2017learning,arunachalam2023optimal,grewal2023efficient,leone2023learning,hangleiter2024bell,aaronson2023efficient,Gluza_2018,mele2024efficient}. Once a state has been recovered by means of quantum state tomography, one can then make predictions for all observables (with bounded operator norm).
{Quantum state tomography~\cite{anshu2023survey,Eisert_2020} and PAC learning~\cite{Aaronson_2007} address fundamentally different tasks, as discussed in Ref.~\cite{Aaronson_2007}. In quantum state tomography, the experimenter has the freedom to choose measurements in order to reconstruct the unknown state. In contrast, PAC learning assumes a fixed set of measurements and focuses on predicting certain properties of the quantum state, rather than directly reconstructing it.}
More specifically, PAC learning of quantum states, as introduced by Aaronson~\cite{Aaronson_2007}, focuses on predicting unseen expectation values based on a limited number of observed expectation values provided in a training set. Aaronson demonstrated that finding a state that accurately matches $O(n)$ expectation values from a given training set suffices to solve the PAC-learning problem. Such a state, consistent with the training data, ensures robust generalization in predicting unseen expectation values originating from the same underlying distribution as the training set. However, finding such a consistent state can be computational difficult depending on the analyzed class of states.
As discussed in the previous section, in the context of ``proper PAC-learning," the objective is to find a hypothesis state that not only matches the training data but also belongs to the same class as the unknown state from which the training data were generated. 

Crucially, with the prior assumption that the unknown state is a fermionic Gaussian state, quantum state tomography can be performed efficiently in sample and time complexity~\cite{Gluza_2018,aaronson2023efficient,ogorman2022fermionic,mele2024efficient}. Specifically, to learn an $n$-qubit fermionic Gaussian state, it is sufficient to measure all the $O(n^2)$ $2$-body correlation functions that constitute the entries of the correlation matrix. However, it is important to note that in the proper PAC-learning setting, access to all $2$-body expectation value estimates is not guaranteed -- the given training set can be composed by few $2$- body, but also higher order correlations.
By showing that proper PAC learning of free-fermionic states is hard with respect to Majorana observables (which is a well-motivated class of observables in the fermionic context), we also refute the idea that PAC learning is simpler than quantum state tomography, which is interesting since the notion of PAC learning of quantum states has been introduced by Aaronson~\cite{Aaronson_2007} as a simplification of tomography. This stands in contrast to stabilizer states and matrix product states, where both tomography and proper PAC learning can be accomplished with time efficiency with respect to `natural' observables~\cite{rocchetto2018stabiliser,yoganathan2019condition}.

{Another related learning framework is property testing of quantum states~\cite{montanaro2018surveyquantumpropertytesting}. In property testing, the goal is to determine whether a given state belongs to a specified class, such as the class of fermionic Gaussian states, or is significantly distant from such a class, based on measurements performed on copies of the state. 
It might also be relevant to emphasize that while our proper PAC-learning setting bears resemblance to the problem of property testing fermionic Gaussian states, it is not the same task. In contrast to property testing, in our proper PAC-learning scenario, the learner cannot choose the measurements to perform; the provided dataset is fixed and not chosen by the learner, and the problem is purely computational.
Remarkably, in stark contrast with proper PAC-learning, the problem of property testing pure of low-rank fermionic Gaussian states is sample and computationally efficient, as shown in Ref.~\cite{Bittel2024testing}. Moreover, property testing of general mixed Gaussian states remain sample and thus computationally inefficient~\cite{Bittel2024testing}. }

{Furthermore, in recent years, the learning framework of \emph{shadow tomography} has gained significant attention~\cite{aaronson2018shadow,Huang_2020,Efficient}, where the objective is to accurately estimate a set of expectation values by performing measurements on the unknown quantum state. Specifically, methods addressing the shadow tomography problem for fermionic systems have been explored~\cite{Zhao_2021,low2022classical,wan2023matchgate,denzler2023learning}. These approaches typically involve collecting samples by applying randomly drawn fermionic Gaussian unitary transformations to the unknown state, with the resulting data enabling the estimation of the quantum state properties, including fermionic correlation functions and the correlation matrix of the state. This, in turn, can provide sufficient information to determine whether the state is free-fermionic or significantly far from being free-fermionic (i.e., a solution to the property testing problem previously mentioned; see Ref.~\cite{Bittel2024testing} for further details).
However, this task fundamentally differs from PAC-learning: in shadow tomography, the learner has the flexibility to choose the measurements performed, whereas in PAC-learning, the learner is constrained to a fixed set of measurements.}

\subsubsection{Further related works}
While the narrative that the ``Gaussian world is simple'' is true for notions of classical simulability~\cite{Terhal_2002,cudby2023gaussian} and full state tomography in trace distance~\cite{Gluza_2018,aaronson2023efficient,ogorman2022fermionic,mele2024efficient}, 
our work further corroborates 
and substantiates the idea that this 
is not quite right for other meaningful tasks. For example, it has been shown that 
the Hartree-Fock method is \texttt{NP}-hard in worst-case complexity \cite{HFNP}, a statement that even holds true for the Hartree-Fock method applied to for translationally invariant systems
\cite{Whitfield_2014}. Also, tasks of
free-fermionic distribution learning~\cite{nietner2023free}
have been identified as being computationally difficult. The latter is another task that can be done efficiently for stabilizer states~\cite{Hinsche_2023}, but -- possibly strikingly -- not for free-fermionic states~\cite{nietner2023free}, breaking the often cited similarity between stabilizer states and free fermions. Our result is also reminiscent to the situation of families of \emph{tensor networks} that can be efficiently stored, yet certain important tasks of optimization~\cite{PhysRevLett.100.250501,PhysRevLett.97.260501} or contraction~\cite{PhysRevLett.98.140506} can be computationally hard or even undecidable~\cite{PhysRevLett.113.160503}.

Concerning PAC-learning of quantum states, Ref.~\cite{yoganathan2019condition} gave sufficient conditions for a class of quantum states to be efficiently PAC-learnable. It was shown that if a class of quantum states is classically simulable and additionally adheres to the so-called \textit{invertibility condition}, then such class can be efficiently PAC-learned. Informally, the invertibility condition requires efficiency in characterizing the states within the ansatz class that adhere all the expectation values of the training set. 
Therefore, as a corollary of our result that the class of fermonic Gaussian states are hard to PAC-learn, it is evident that fermionic Gaussian states do violate the invertibility condition.  
In contrast, both stabilizer states~\cite{rocchetto2018stabiliser} and matrix product states~\cite{yoganathan2019condition} do satisfy such invertibility conditions and, as such, they can be efficiently PAC-learned. 
Notions of PAC-learning of quantum circuits have also been defined~\cite{Caro_2020}. Ref.~\cite{Liang_2023} has shown that proper PAC-learning of Clifford circuits is time-efficient if and only if \texttt{RP}=\texttt{NP} (i.e., implying a collapse of the polynomial hierarchy). 

{The NP-hardness of proper PAC learning has been extensively studied in classical learning theory, particularly in the context of Boolean functions~\cite{KearnsVALIANT, PittValiant1988} in the 1980s and 1990s. In particular, properly learning Conjunctive Normal Form (CNF) and Disjunctive Normal Form (DNF) formulas—conjunctions of disjunctions of clauses of literals and disjunctions of conjunctions of terms of literals, respectively—has been shown to be inefficient, unless \( \text{R} = \text{NP} \)~\cite{PittValiant1988}. Specifically, this was proven for learning \( k \)-clause CNF (the class of CNF formulas with at most $k$ clauses) and \( k \)-term DNF (the class of DNF formulas with at most $k$ terms) for \( k \geq 2 \).}

\subsubsection{Organization of our work}
This work is structured as follows: Section~\ref{sec:setup} provides a concise introduction to the key definitions. In Section~\ref{sec:overview}, we present an informal summary of our results, offering simplified proof sketches. Section~\ref{sec:discussion} focuses on our conclusions and explores avenues for future research. In Section~\ref{sec:preliminaries}, we offer a comprehensive introduction to the methodologies employed in our study, covering topics such as fermionic Gaussian states, the 3-SAT problem, and the basics of PAC-learning. Finally, the more technically oriented Section~\ref{sec:proofs} delves into the detailed proofs of our results.

\subsection{Setup}
\label{sec:setup}
{\bf Fermionic Gaussian states.} We consider a system formed by $n$-fermionic modes or, without loss of generality, an $n$-qubits system. Majorana operators are $2n$ operators which can be defined in terms of their anti-commutation relation $\{\gamma_\mu,\gamma_\nu\}=2\delta_{\mu,\nu} I$,
for all $\mu,\nu \in [2n]$, where $[2n]:=\{1,\dots,2n\}$. They can also defined in terms of qubits-Pauli operators as 
\begin{equation}
\gamma_{2 k-1}=(\prod_{j=1}^{k-1} Z_j) X_k\,,\quad \, \gamma_{2 k}=(\prod_{j=1}^{k-1} Z_j) Y_k
\end{equation}
for $k \in [n]$, through the Jordan-Wigner transformation. They are Hermitian, traceless, square to the identity, and different Majorana operators anti-commute. Moreover, the ordered product of them form an orthogonal basis with respect to the Hilbert Schmidt scalar product for the space of linear operators. Given a quantum state $\rho$, we denote as its \emph{$k$-th order correlation functions} the expectation values 
\begin{equation}
\langle\gamma_S\rangle_{\rho}\coloneqq i^{-|S|/2}\tr(\gamma_{S}\rho), 
\end{equation}
where $\gamma_{S}\coloneqq\gamma_{\mu_1}\cdots\gamma_{\mu_{|S|}}$ with $S\coloneqq \{\mu_1,\ldots, \mu_{|S|}\}\subset[2n]$ and $1 \le \mu_{1}<\dots <\mu_{|S|}\le 2n$. An operator is referred to as a $k$-body \emph{Majorana operator}  (or product) if it can be expressed as the product of $k$ Majorana operators.  Let $\mathcal{S}$ be a set of ordered indices subsets $S\subset[2n]$. We identify $\mathcal{S}$ with a set of Majorana product operators via the mapping 
\begin{equation}
\mathcal{S}\ni S=\{\mu_1,\ldots,\mu_{|S|}\}\mapsto \gamma_{S}=\gamma_{\mu_1}\cdots\gamma_{\mu_{|S|}}.
\end{equation}

Given a quantum state $\rho$, its 
\emph{correlation matrix} $\Gamma(\rho)$ is defined as the $2n \times 2n$ anti-symmetric real matrix with elements 
\begin{equation}
[\Gamma(\rho)]_{j,k}:=-\frac{i}{2}\Tr\left(\left[\gamma_j,\gamma_k\right]\rho\right)= \langle\gamma_{j}\gamma_k\rangle_{\rho},
\end{equation}
for all $j\le k\in [2n]$.
We say that a matrix is a valid correlation matrix if it is real and anti-symmetric, with eigenvalues in absolute value less than or equal to one. In particular, the correlation matrix of any (general) state is a valid correlation matrix.
Correlation matrices are particularly relevant because there is a bijection between fermionic Gaussian states and valid correlation matrices, which is to say that fermionic Gaussian states are uniquely determined by its $2$-body Majorana products encoded in $\Gamma(\rho)$. In particular, given a Gaussian state $\rho$, higher order correlation functions are given by Wick's theorem~\cite{Surace_2022} as 
\begin{equation}
\langle\gamma_{S}\rangle_{\rho}  =\operatorname{Pf}(\Gamma(\rho)\rvert_{S}), 
\end{equation}
for $\Gamma(\rho)\rvert_{S}$ being the restriction of the matrix $\Gamma$ to the rows and columns corresponding to elements in $S$. Here, $\operatorname{Pf}(\Gamma)$ indicates the Pfaffian of the matrix $\Gamma$, which, for anti-symmetric matrices, reads $\operatorname{Pf}(\Gamma)^2=\det(\Gamma)$ (see Section~\ref{sec:preliminaries} for more details).

\medskip

{\bf 3-SAT decision problem.} The task of determining whether an object possesses a specific property (for instance, verifying if the data align with a fermionic Gaussian state) is framed as a decision problem~\cite{arora_barak_2009}.
The complexity class \texttt{NP} encompasses decision problems efficiently verifiable in polynomial time. A problem is \texttt{NP}-hard if it is at least as computationally hard as the hardest problems in \texttt{NP}. If a problem is both in \texttt{NP} and \texttt{NP}-hard, it is classified as \texttt{NP}-complete.
In our hardness proof, we establish a reduction from our specific problem to the well-known \texttt{NP}-complete 3-SAT problem. In 3-SAT, Boolean formulas consist of clauses with three literals, which can be either Boolean variables or negations (e.g., \(x_1 \vee \neg x_2 \vee x_3\)). The objective is to find truth assignments that satisfy the formula. The problem can be formally defined as follows.

\begin{problem}[$3$-SAT]
\label{prob:3satMain}
Problem Statement:
\begin{itemize}
    \item \textbf{Instance:} Given a list of indices \(\vec i\in [\nx]^{3,n_c}\) and a negation table \(s\in \{\cdot, \lnot\}^{3,n_c}\).
    
    \item \textbf{Question:} Is there a Boolean variable assignment \(x\in\{0,1\}^{\nx}\) that satisfies all of the following clauses?
    \begin{align}
        \forall c\in[n_c]:\quad s_{1,c} x_{i_1(c)}\lor s_{2,c} x_{i_2(c)}\lor s_{3,c} x_{i_3(c)}.
    \end{align}
\end{itemize}
\end{problem}
Here, $[\nx]^{3,n_c}$ means that we take $n_c$ subsets of $3$ elements belonging to $[\nx]$.
To establish the \texttt{NP}-hardness of (proper) PAC learning of free-fermionic states, we proceed in two steps. We first show that $\class{Gaussian\,Consistency}$ decision problem, defined subsequently, reduces to the PAC-learning problem via standard complexity arguments. Second, we establish a polynomial-time reduction from the 3-SAT problem to the $\class{Gaussian\,Consistency}$ problem. This implies that efficient proper PAC learning of free-fermionic states would imply efficiently solving 3-SAT, which, in turn, would lead to the (widely not-believed) conclusion \(\class{P}=\class{NP}\).
In the problem statements presented below, the input data of the problems are represented by rational numbers, which possess finite descriptions and thus serve as viable inputs for computer programs. Allowing for irrational numbers would only make the problem more computational challenging.

\subsection{Overview of the main results} \label{sec:overview}
In this section, we unveil the main result of this work and elucidate the rationale and implications underpinning the proof. Let $\mathcal{S}$ denote a set of ordered indices $\mathcal{S}\ni S\subset[2n]$ encoding Majorana product operators, i.e., $\gamma_S$ for $S\in\mathcal{S}$. We refer to a \textit{training set} $\mathcal{T}=\{(S,\langle\gamma_S\rangle_T)\}_{S\in \mathcal{S}}$, where $\langle\gamma_S\rangle_T$ corresponds to the estimation of the \textit{true} expectation value $\langle\gamma_S\rangle_{\rho}$, given within an accuracy $\varepsilon$, i.e., $|\langle\gamma_S\rangle_{\rho}-\langle\gamma_S\rangle_{T}|\le\varepsilon$ for every $S\in\mathcal{S}$.  We formally define the \emph{proper PAC-learning problem for fermionic Gaussian states} as follows:

\begin{problem}[Proper PAC-learning of Gaussian states]\label{prob:pac_approx} 
Let $\rho$ be a Gaussian state and $\langle\gamma_S\rangle_{\rho}$ be expectation values of Majorana correlation functions for $S\in\mathcal{S}$.  Problem statement:
\begin{itemize}
    \item \textbf{Instance:} A training set $\mathcal{T}=\{(S,\langle\gamma_S\rangle_T)\}_{S\in \mathcal{S}}$ such that 
    {$\left<\gamma_S\right>_T=\langle\gamma_S\rangle_{\rho}$
    for every $S \in \mathcal{S}$.}
    \item \textbf{Output:} Provide a fermionic Gaussian state $\sigma$, characterized by its correlation matrix $\Gamma(\sigma)$, that fulfills the constraints 
    \begin{align}
    |\langle\gamma_S\rangle_{T}-\langle\gamma_S\rangle_{\sigma}| \le \varepsilon,\, \forall S\in \mathcal{S}.
    \end{align}
\end{itemize}
\end{problem}

{Two important remarks are in order regarding the definition of Problem~\ref{prob:pac_approx}. First, we consider the \textit{worst-case} output, meaning that the quantum state must match all expectation values within the specified accuracy parameter. Another commonly used notion, also adopted in Aaronson's original work~\cite{Aaronson_2007}, is to quantify the error using the \textit{average error} over the training set. Second, the definition of Problem~\ref{prob:pac_approx} assumes that the expectation values $\langle\gamma_{S}\rangle_{T}$ in the training set originate from a free fermionic state and, crucially, are exact with no error. However, in practical settings, due to shot noise statistics, the expectation values in the training set— even if derived from a free fermionic state $\rho$—will only be $\varepsilon$-close to the true values. Nevertheless, since the hardness of noisy learning follows from the hardness of the exact version, our proof establishes hardness for both cases.  }
{We emphasize once again that Problem~\ref{prob:pac_approx} requires finding a free fermionic state that is consistent with the expectation values in the training set. This is referred to as a \textit{proper} PAC learning problem, as discussed above.  
} 
Let us now present the $\class{Gaussian\,Consistency}$ problem for which we show the $\class{NP}$-completeness. Unlike the \emph{proper PAC-learning Problem}~\ref{prob:pac_approx}, for the $\class{Gaussian\,Consistency}$ problem we are provided with a training set of Majorana product expectation values that arise from a general quantum state $\rho$, not necessarily Gaussian. The question is whether there exists a fermionic Gaussian state $\sigma$ capable of accurately matching this data.
We formalize this problem, as a decision problem, as follows:
\begin{problem}[$\class{Gaussian\,Consistency}$ problem]\label{prob:GaussianConsMAIN} 
Problem statement:
\begin{itemize}
    \item \textbf{Instance:} A training set
    $\mathcal{T}=\{(S,\langle\gamma_S\rangle_T)\}_{S\in \mathcal{S}}$ and an accuracy parameter $\varepsilon>0$.
    \item \textbf{Promise:} 
    {The training set is obtained from a physical state $\rho$, where $\left<\gamma_{S}\right>_{T} = \langle\gamma_S\rangle_{\rho}$.}
    \item \textbf{Question:} Does there exist a Gaussian state $\sigma$ that matches the expectation values in the training set up to accuracy $\varepsilon$:
    \be
            |\langle\gamma_S\rangle_{T}-\langle\gamma_S\rangle_{\sigma}| \le \varepsilon,\quad \forall S \in \mathcal{S}\,\text{?}
    \ee
\end{itemize}
\end{problem}
We recall that the $2$-point (or $2$-body) correlation functions of a state constitute its correlation matrix entries, while higher-order functions introduce Pfaffian constraints because of Wick's theorem (see Sec.~\ref{sec:preliminaries} for details). Consequently, our task of identifying a fermionic Gaussian state consistent with the given data can be framed as a matrix completion problem for an anti-symmetric real matrix that involves satisfying Pfaffian constraints on submatrices along with linear constraints on the eigenvalues. The intuition behind our \texttt{NP}-hardness result is based on the fact that Pfaffian constraints are nonlinear and make the matrix completion problem more challenging than linear constraints alone. In contrast, in the stabilizer states case~\cite{rocchetto2018stabiliser}, the constraints imposed by the observed Pauli expectation values are all linear: this is the high-level intuition that explains why PAC learning of stabilizer states with respect to Pauli (Majorana) observables can be solved efficiently~\cite{rocchetto2018stabiliser}, while for Gaussian fermionic states it cannot.   
By standard argument of complexity theory, we can demonstrate that Problem~\ref{prob:pac_approx} is at least as hard as the Gaussian consistency Problem \ref{prob:GaussianConsMAIN}.

\begin{lemma}[Proper PAC learning of Gaussian states and the 
$\class{Gaussian\,Consistency}$ problem]\label{lem:equivalence}
Problem~\ref{prob:pac_approx} (proper PAC-learning) is at least as hard as Problem~\ref{prob:GaussianConsMAIN} (Gaussian consistency problem).
\begin{proof}
We can insert an instance of Problem~\ref{prob:GaussianConsMAIN} to an algorithm solving Problem~\ref{prob:pac_approx}. If the algorithm returns a valid correlation matrix such that it allows to approximate all the expectation values in the training set with accuracy $\epsilon$, we accept. If the algorithm fails or returns an invalid correlation matrix we reject. Therefore, there is a single call Cook reduction~\cite{arora_barak_2009} between the two problems. 
\end{proof}
\end{lemma}

To establish \texttt{NP}-completeness of Problem~\ref{prob:GaussianConsMAIN}, and thus \texttt{NP}-hardness of Problem~\ref{prob:pac_approx} through Lemma~\ref{lem:equivalence}, we seek to perform a reduction of the $3$-SAT Problem~\ref{prob:3satMain} to Problem~\ref{prob:GaussianConsMAIN}. Given a $3$-SAT instance, our objective is to reformat the inputs to align with the requirements of the $\class{Gaussian \,Consistency}$ Problem~\ref{prob:GaussianConsMAIN}, such that this transformation enables us to deduce the corresponding output for the $3$-SAT problem based on the output result of the 
$\class{Gaussian\,Consistency}$ problem.

\begin{theorem}[\texttt{NP}-completeness of \(\class{Gaussian\, Consistency}\)]\label{th:1main}
    The \(\class{Gaussian\,Consistency}\) Problem \ref{prob:GaussianConsMAIN} is \(\class{NP}\)-complete, even if we restrict the dataset to $k$-body Majorana operators with $k\le 4$ and constant accuracy \(\varepsilon=O(1)\).
\end{theorem}
\begin{proof}[Proof sketch]
    We present the proof sketch for the easiest case, i.e., for \(\varepsilon=0\) and $k\le 6$ body correlation functions. We refer to Section~\ref{sec:proofs} for the more general case of $k\le 4$ and the analysis of finite accuracy \(\varepsilon>0\). 
    
    Consider a system of \(2\nx\) fermionic modes associated with a \(4\nx \times 4 \nx\) correlation matrix \(\Gamma\), where \(\nx\) is the number of \(3\)-SAT variables.  
    To encode \(3\)-SAT clauses into our problem, we impose constraints on \(\Gamma\). Specifically, we require the following block diagonal structure $
    \Gamma=\Gamma_1\oplus \cdots\oplus \Gamma_{\nx}$,
    where \(\Gamma_j\) is defined as
    \begin{align}
        \Gamma_j = \frac{1}{\Phi}\begin{pmatrix}
            0 & x_j^0 & 0 & x_j^1 \\
            -x_j^0 & 0 & 1 & 0 \\
            0 & -1 & 0 & 1 \\
            -x_j^1 & 0 & -1 & 0
        \end{pmatrix},
        \label{eq:6corrMAIN}
    \end{align}
    where \(\Phi\) is a constant, and \(x_j^0, x_j^1 \in [0,1]\) are parameters related to the \(3\)-SAT variables. Throughout the remainder of this proof, we adopt the notation \(\gamma_{j,k}:=\gamma_{4(j-1)+k}\), for all \(j\in [\nx]\), \(k\in [4]\). Additionally, we demand that 
    \begin{equation}
        \left<\gamma_{j
        ,1}\gamma_{j,2}\gamma_{j,3}\gamma_{j,4}\right>=\frac{1}{\Phi^2}
        \end{equation}
    for all \(j\in[\nx]\). {Next, we employ Wick's theorem~\cite{Surace_2022}, which relates higher-order Majorana expectation values of fermionic Gaussian states to their quadratic Majorana expectation values (see Lemma~\ref{le:Wick} for more details).} Using Wick's theorem, we deduce that \(x_j^1=1-x_j^0\). 
    We are now prepared to encode general \(3\)-SAT clauses. Let us consider the clause $
    s_1 x_{i} \lor s_2 x_{j} \lor s_3 x_{k}$, as in Problem~\ref{prob:3satMain}.
    In the following, we use \(s_{j} \in \{0,1\}\) such that \(s_{j} = 1\) if \(s_{j} \mapsto \lnot\). We impose the following six-point correlators to be zero,
    \begin{align}
        &\left<\gamma_{i,1}\gamma_{i,2+2s_1}\gamma_{j,1}\gamma_{j,2+2s_2}\gamma_{k,1}\gamma_{k,2+2s_3}\right> = 0.
    \end{align}
Again by Wick's theorem, this corresponds to the encoded clause $ x_{i}^{s_1}x_{j}^{s_2}x_{k}^{s_3} = 0$.
    Therefore, by solving this constrained matrix completion problem, we can find a solution for \(3\)-SAT. 
    The problem is also in $\class{NP}$ because given a correlation matrix and a training set, we can efficiently check if this correlation matrix is consistent with the training set. Hence, \(\class{Gaussian\, Consistency}\) is \(\class{NP}\)-complete. 
\end{proof}
 It is worth noting that, in the proof sketch of Theorem~\ref{th:1main}, the utilized correlation functions are confined to $k\leq 6$-point correlations. However, demonstrating a similar result with correlation functions limited to only up to $4$-points enhances the robustness of the conclusion and aligns better with practical experimental constraints. Additionally, in Section~\ref{sec:physicality}, we demonstrate that the correlation matrices employed for the 3-SAT reduction, such as the one in Eq.~\eqref{eq:6corrMAIN}, genuinely correspond to a physical state. 
 The following corollary formally establishes the \texttt{NP}-hardness of Problem~\ref{prob:pac_approx}.

\begin{corollary}[The proper PAC-learning of fermionic Gaussian states is \texttt{NP}-hard]
Proper PAC-learning problem~\ref{prob:pac_approx} of fermionic Gaussian states is \texttt{NP}-hard, even if we restrict the dataset to $k$-body Majorana operators with $k\le 4$ and accuracy $\varepsilon=O(1)$.
\begin{proof}
Thanks to Lemma~\ref{lem:equivalence}, we know that Problem~\ref{prob:pac_approx} is at least as Hard as Problem~\ref{prob:GaussianConsMAIN}, which is \texttt{NP}-complete due to Theorem~\ref{th:1main}. Hence,  Problem~\ref{prob:pac_approx} is \texttt{NP}-hard (under Cook reduction).
\end{proof}
\end{corollary}

The careful reader may wonder whether an analogous no-go result can be provided in the case of $k\leq 2$ Majorana product operators. However, the answer is trivially negative: indeed, recalling that the correlation matrix $\Gamma(\rho)$ is comprised of $k=2$ Majorana products, in this case, it would be sufficient for the learner to complete the correlation matrix without any Pfaffian constraints which can be solved by a SDP in $\mathrm{poly}(n)$ running time. As outlined in the proof sketch of Theorem~\ref{th:1main}, the core of the hardness result lies in the difficulty of satisfying Pfaffian constraints that are naturally nonlinear constraints.

Before concluding the section, it is important to highlight that our prior results effortlessly generalize to the situation where we inquire about the existence of a \textit{pure} Gaussian state that accurately matches the provided data. This stems from the fact~\cite{Windt_2021} that, for any valid correlation matrix associated with an $n$-mode mixed fermionic Gaussian state, we can always construct a correlation matrix for $2n$-modes corresponding to a pure state. This pure state exhibits the property that its reduced state corresponds to the initial state, when solely considering the first $n$-modes and tracing out the rest.

\subsection{Discussion and open questions}
\label{sec:discussion}
Our investigation reveals that determining the existence of a fermionic Gaussian state consistent with a given dataset of few $2$- and $4$-point Majorana expectation values is NP-complete, even when allowing for constant errors in matching the provided expectation values. This discovery marks the first instance of classically simulable class of quantum states identified as not efficiently proper PAC-learnable, signifying a significant advancement in quantum learning theory and prompting several engaging research directions discussed subsequently.

We have shown that proper PAC-learning is hard with respect Majorana expectation values, however this does not necessarily extend to the realm of improper PAC-learning—where the hypothesis state is not constrained to be Gaussian. 
Nevertheless, the setup of proper learning is convenient: when creating a PAC learning algorithm for a given class of states, a learner usually starts with an ansatz class of quantum states required to match the expectation values provided in the training set. A natural initial step is to choose the ansatz class within the concept class. In this way, the learner is automatically ensured that there exists at least one quantum state satisfying all the training data.  
Additionally, searching for hypothesis states in the concept class is convenient because one must be able to compute expectation values efficiently for future classical predictions—a condition automatically fulfilled when the hypothesis state belongs to a classically simulable ansatz class, such as fermionic Gaussian states.  
{While the setting of proper PAC learning is natural, the question of whether improper learning offers advantages for PAC-learning fermionic Gaussian states is also interesting and it remains open. Unfortunately, our current proof techniques are unable to capture this case, primarily because they rely heavily on Wick's theorem (Lemma~\ref{le:Wick}). However, we expect that an alternative approach could be developed for improper PAC learning and leave this as an open direction for future work.}  
It would be also interesting to understand if there are other meaningful classes of observables expectation values, different from Majorana correlation functions, for which proper PAC-learning of fermionic Gaussian states can be solved in polynomial time. 

Furthermore, we note that the efficiency of PAC-learning states that slightly deviate from stabilizer states (which can be PAC-learned efficiently~\cite{rocchetto2018stabiliser}) remains an open question. For example, $t$-doped stabilizer states~\cite{leone2023learning22PUBL1,leone2023learning22PUBL2} constitute an efficiently simulable class for $t=O(\log n)$~\cite{gottesman1998heisenberg}, and their PAC-learning efficiency is yet to be determined. Similarly, states with small stabilizer rank~\cite{Bravyi_2016} or small stabilizer extent~\cite{Bravyi_2019} are classically simulable, but little is known about the efficiency of learning tasks for these classes (both in the tomography and the PAC-learning framework). Moreover, it would be explore if hypergraph states~\cite{Rossi_2013} can be efficiently PAC-learned.

\section{Preliminaries}\label{sec:preliminaries}
This section comprehensively lays out the concepts needed to understand our results.

\subsection{Fermionic Gaussian states}
In this subsection, we provide definitions and lemmas on fermionic Gaussian states, which are useful for deriving our results.  
We consider a system composed of $n$ modes with creation and annihilation operators $a_{j}^{\dagger}$ and $a_{j}$, with $j\in[n]$, satisfying the canonical anti-commutation relations~\cite{Surace_2022},
\begin{align}
    \{a_{k}, a_{l}\}=0, \quad \{a_{k}, a_{l}^{\dagger}\}=\delta_{k,l}.
\end{align}
From this, we introduce the definition of Majorana operators. 
\begin{definition}[Majorana operators]
We can define the Majorana operators as
\begin{align}
    \gamma_{2k-1}  \coloneqq a_k + a_k^\dagger, \quad\quad
    \gamma_{2k}  \coloneqq i(a_k^\dagger - a_k).
\end{align}
\end{definition}

Majorana operators satisfy the anti-commutation relations 
\begin{align}
    \{\gamma_\mu,\gamma_\nu\}=2\delta_{\mu,\nu} I \quad \text{for all} \quad \mu,\nu \in [2n]
\end{align}
and are orthogonal with respect to the Hilbert-Schmidt inner product. It holds for $\mu \in [2n]$
that
\begin{align}
    \gamma_\mu=\gamma^\dagger_\mu, \quad \Tr(\gamma_\mu)=0, \quad \gamma^2_\mu=I.
\end{align}

Majorana operators can be represented also in terms of the Pauli basis as follows.
\begin{definition}[Jordan Wigner Transform]
For each \(k \in \left[n\right]\), Majorana operators can be represented by Pauli operators on $n$ qubits via the Jordan-Wigner transform 
\begin{align}
    \gamma_{2 k-1}=\left(\prod_{j=1}^{k-1} Z_j\right) X_k, \quad \gamma_{2 k}=\left(\prod_{j=1}^{k-1} Z_k\right) Y_k.
\end{align}
\label{def:majo}
\end{definition}
As can be readily verified, that these operators have all the desired properties. 
\begin{definition}[Majorana products]
    Let \(S\) be the set \(S := \{\mu_1,\dots,\mu_{|S|}\} \subseteq [2n]\) with \(1\le\mu_1 <\dots < \mu_{|S|}\le 2n \). 
    We define  $$\gamma_S:=\gamma_{\mu_1}\cdots\gamma_{\mu_{|S|}}$$ if $S\neq \emptyset$ and  $\gamma_{\emptyset}=I$ otherwise.
\end{definition}
It is worth noting that the number of different sets \(S\in[2n]\) and, hence, Majorana products is \(4^n\).
For any set $S,S^\prime \subseteq [2n]$, Majorana products are orthogonal with respect to the Hilbert Schmidt product. Hence, they form a basis for \(\mathcal{L}(\mathcal{H}_n)\). For convenience, it will be useful to denote the expectation values of Majorana products as
$$
\langle\gamma_{S}\rangle_\rho\coloneqq i^{-|S|/2}\tr(\gamma_S\rho)\,.$$

\begin{definition}[Correlation matrix]
    Given a (general) quantum state $\rho$, we define its \emph{correlation matrix} $\Gamma(\rho)$ as
    \begin{align}
        [\Gamma(\rho)]_{j,k}:=-\frac{i}{2}\Tr\left(\left[\gamma_j,\gamma_k\right]\rho\right)
    \end{align}
    where $j,k\in[2n]$. Moreover, it is useful to define the shorthand $ \langle\gamma_{j}\gamma_k\rangle_{\rho}:=[\Gamma(\rho)]_{j,k}$.
\end{definition}
The correlation matrix of any state is real and anti-symmetric, thus it has eigenvalues in pairs of the form $\pm i \lambda_j$ for $j\in[2n]$, where $\lambda_j$ are real numbers such that $|\lambda_j|\le 1$.

\begin{definition}[Fermionic Gaussian states]
    Fermionic Gaussian states are the set of states defined as (the closure of) Gibbs states of quadratic Hamiltonians in the Majorana operators, i.e.,
    \begin{align}
        H&=i \sum_{\nu,\mu\in [2n]}h_{\nu\mu}\gamma_{\nu}\gamma_{\mu},\\
        \rho&=\frac{\exp(-H)}{\Tr(\exp(-H))}.
    \end{align}
    Here, $h\in \mathbb{R}^{2n \times 2n}$ is anti-symmetric $h^T=-h$.
\end{definition}
From the previous definition, it can be shown that $\Gamma(\rho)=\tan(2h)$, where $\tan(\cdot)$ is the matrix function of the tangent. Since the tangent is injective for anti-symmetric matrices, any Gaussian state is uniquely determined by its correlation matrix.
Thus, by specifying a valid correlation matrix (i.e., real, anti-symmetric, with eigenvalues smaller than one), we uniquely specify a Gaussian state. Vice versa, having a Gaussian state uniquely defines a correlation matrix.

\begin{definition}[$k$-point correlation function]
Given a quantum state $\rho$ and an operator $\gamma_S := \gamma_{\mu_1}\cdots\gamma_{\mu_{k}}$, where $S = \{\mu_1, \dots, \mu_{k}\} \subseteq [2n]$ with $1 \leq \mu_1 < \dots < \mu_{k} \leq 2n$, we define $\Tr(\gamma_S\rho)$ as a $k$-point correlation function.
\end{definition}
As in the formulation of the PAC-learning problem, the training set is comprised of a set of Majorana product operators, it is worth defining them formally in the following definition. 

\begin{definition}[Set of $k$-point correlation functions] Let $\mathcal{S}$ be a set of ordered indices $S\subset[2n]$. We identify $\mathcal{S}$ with a set of Majorana product operators via the mapping
\be
\mathcal{S}\ni S=\{\mu_1,\ldots,\mu_{|S|}\}\mapsto \gamma_{S}=\gamma_{\mu_1}\cdots\gamma_{\mu_{|S|}}.
\ee
\end{definition}

In our analysis it will be crucial the notion of 
the \emph{Pfaffian} of a matrix.

\begin{definition}[Pfaffian of a matrix]
Let $C$ be a $2n\times 2n$ anti-symmetric matrix. Its Pfaffian is defined as
\begin{align}
\operatorname{Pf}(C)=\frac{1}{2^n n !} \sum_{\sigma \in S_{2 n}} \operatorname{sgn}(\sigma) \prod_{i=1}^n C_{\sigma(2 i-1), \sigma(2 i)},
\end{align}
where $S_{2 n}$ is the symmetric group of order $(2 n) !$ and $\operatorname{sgn}(\sigma)$ is the signature of $\sigma$. 
The Pfaffian of an $m\times m$ anti-symetric matrix with $m$ odd is defined to be zero.
\end{definition}
Well-known properties are the following. For any matrix \(Q\), we have \(\operatorname{Pf}(QCQ^T) = \det(Q) \operatorname{Pf}(C)\) and $\operatorname{Pf}(\lambda C) = \lambda^{n} 
 \operatorname{Pf}(C)$, where $C$ is a $2n\times 2n$ anti-symmetric matrix and $\lambda \in \mathbb{C}$.
Moreover, it holds that \(\operatorname{Pf}(C)^2 = \det(C)\) (note that this is consistent with the fact that the Pfaffian of an odd anti-symmetric matrix is defined to be zero, since the determinant of an odd anti-symmetric matrix is zero).
Another useful identity is
\begin{align}
\operatorname{Pf}\left(\bigoplus_{j = 1}^{n} \begin{pmatrix} 0 &  \lambda_j \\ -\lambda_j & 0 \end{pmatrix}\right)= \prod^n_{j=1} \lambda_j.
\end{align}
Now we recall the well-known Wick's theorem, which states that the any correlation function (i.e.,  Majorana product expectation value) over a fermionic Gaussian state can be computed efficiently given access to its correlation matrix. 
\begin{lemma}[Wick's Theorem \cite{Surace_2022}]
\label{le:Wick}
    Let $\rho$ be a fermionic Gaussian state with the associated correlation matrix $\Gamma(\rho)$. Then, we have
    \begin{align}
        \left<\gamma_S\right>_\rho\coloneqq i^{-|S|/2}\Tr(\gamma_S \rho) = \operatorname{Pf}( \Gamma(\rho)\rvert_{S}),
    \end{align}
    where $\gamma_S = \gamma_{\mu_1}\cdots\gamma_{\mu_{|S|}}$, and $S = \{\mu_1, \dots, \mu_{|S|}\} \subseteq [2n]$ with $1 \leq \mu_1 < \dots < \mu_{|S|} \leq 2n$, while $\Gamma(\rho)\rvert_{S}$ is the restriction of the matrix $\Gamma(\rho)$ to the rows and columns corresponding to elements in $S$. 
\end{lemma}
Note that, for any $S\subseteq [2n]$, the restriction of the correlation matrix $\Gamma\rvert_{S}$ is still anti-symmetric, Therefore,  its Pfaffian is well-defined. 
The Pfaffian of $\Gamma\rvert_{S}$ can be computed efficiently in time $\Theta(|S|^{3})$.
In our discussion, it will be useful the following lemma.

\begin{lemma}[Purification of a mixed Gaussian state \cite{Windt_2021}] 
\label{le:purification}
A fermionic Gaussian mixed state of $n$-modes can be always purified to a fermionic Gaussian pure state of $2n$-modes.
\end{lemma}

\subsection{The 3-SAT problem and \texttt{NP}-hardness}
For a comprehensive overview of complexity theory, we recommend consulting standard references like~\cite{arora_barak_2009}. In this section, we will outline notions related to the 3-SAT problem and a common approach employed to establish \texttt{NP}-hardness.
The 3-SAT problem is a fundamental computational problem in complexity theory. It belongs to the class of problems known as \texttt{NP}-complete, which implies that it is both in the complexity class \texttt{NP}-(\emph{non-deterministic polynomial time}) and is at least as hard as the hardest problems in \texttt{NP}.
In 3-SAT, we are given a Boolean formula consisting of a \emph{conjunction} (AND) of clauses. Each clause is a \emph{disjunction} (OR) of exactly three literals. A literal can be either a variable or its negation. For example, a clause may take the form \( (x_1 \vee \neg x_2 \vee x_3) \), where \(x_i\) represents a Boolean variable and \(\neg\) denotes negation.
The goal of the 3-SAT problem is to determine if there exists an assignment of truth values (true or false) to the variables such that the entire Boolean formula evaluates to true. This assignment is often referred to as a satisfying assignment.
Given a 3-SAT formula \( \phi \) with \(n\) variables and \(m\) clauses, the question is whether there exists an assignment of truth values \(x_1, x_2, \ldots, x_n\) such that \( \phi(x_1, x_2, \ldots, x_n) \) evaluates to true.
Subsequently, we provide the rigorous definition of the problem.
\begin{problem}[$3$-SAT]\label{prob:3sat} Problem:
    \begin{itemize}
        \item Input: A list of indices $\vec i\in [\nx]^{3,n_c}$ and a negation table $s\in \{\cdot, \lnot\}^{3,n_c}$.
        \item Output: Determine, if a Boolean variable $x\in\{0,1\}^{\nx}$ exists which satisfies all clauses.
        \begin{align}
            \forall c\in[n_c]:\quad s_{1,c} x_{i_1(c)}\lor s_{2,c} x_{i_2(c)}\lor s_{3,c} x_{i_3(c)}.
        \end{align}
    \end{itemize}
    \end{problem}
Here, $[\nx]^{3,n_c}$ means that we take $n_c$ subsets of $3$ elements belonging to $[\nx]$.
To establish the \texttt{NP}-hardness of a problem, a commonly employed approach involves demonstrating a polynomial-time reduction from $3$-SAT to the specific problem at hand. This process entails the transformation of $3$-SAT instances into instances of the target problem in a manner that a solution to the transformed instance directly implies a solution to the original $3$-SAT instance.
The key insight lies in recognizing that efficient resolution of the target problem would also imply an efficient resolution of $3$-SAT, which is known to be \texttt{NP}-hard. To show the legitimacy of the reduction, one must prove that it operates in polynomial time and satisfies the conditions of ``completeness" and ``soundness".
``Completeness" asserts that if a solution exists for $3$-SAT, then a solution also exists for the transformed version of the problem of interest. 
``Soundness" states that if a solution does not exist for $3$-SAT, then there is no solution also for the transformed version of the problem of interest.
Thus, the \texttt{NP}-hardness of 3-SAT serves as a reference point for proving the computational intractability of other problems, including the one currently under consideration in our work.

\subsection{Background on probably approximately correct (PAC) learning of quantum states}
\label{sec:PAC}

A pivotal task in quantum information science is the one of characterizing quantum states through observations, such as outcomes from measurements performed on the system. Various learning paradigms have been introduced to address this challenge. One such approach is full quantum state tomography~\cite{anshu2023survey}, which aims to furnish a classical representation of the density matrices that approximate the unknown state. However, it is well-established that in the worst-case scenario, full state tomography necessitates a sample size that scales exponentially with the number of qubits in the system~\cite{Haah_2017}.
In a seminal work by Aaronson \cite{Aaronson_2007}, an alternative paradigm for learning quantum states, termed \emph{probably approximately correct} (PAC)-learning of quantum states, has been introduced, based on the 
familiar concept of PAC learning in computational learning theory. In this framework, the objective is to predict the expectation values of \emph{positive operator-valued measure} (POVM) elements. This prediction is based on a training set comprised of pairs of POVM elements and their corresponding expectation values, drawn from a distribution over the POVM elements. The aim is to accurately predict the expectation values of unseen POVM elements that share the same underlying distribution.

Remarkably, Aaronson has demonstrated a surprising result: The entire class of $n$-qubit quantum states can be PAC-learned using only $O(n)$ samples from such a training set. However, while Aaronson's result addresses sample efficiency, it does not address the question of whether time-efficient PAC-learning of specific classes of quantum states can be performed. Subsequent research has made further advancements in this context. For instance, it has been demonstrated that classes such as stabilizer states \cite{rocchetto2018stabiliser} and states with low entanglement (low Schmidt rank) \cite{yoganathan2019condition} can be PAC-learned time-efficiently. Furthermore, conditions have been identified that imply time-efficient PAC-learnability \cite{yoganathan2019condition}.
These classes are noteworthy examples of quantum states that can be also efficiently simulated classically. 
In Ref.~\cite{Aaronson_2007}, Aaronson has extended the PAC learning framework from classical learning theory to the domain of quantum states. For a detailed mapping between the standard notations in classical learning theory and PAC learning of quantum states, we refer the interested reader to 
Ref.~\cite{rocchetto2018stabiliser}.
The problem of PAC learning of quantum states can be formulated as follows. We are given a set $T$ comprising $m$ two-outcome POVM elements $\{E^{(1)}_i\}_{i\in[m]}$, along with their associated expectation values $\{\operatorname{Tr}(\rho E^{(1)}_i )\}_{i\in[m]}$, where $\rho$ represents the unknown quantum state. We represent this set as 
\begin{equation}
T := \{( E^{(1)}_i, \operatorname{Tr}(\rho E^{(1)}_i ) )\}_{i\in[m]},
\end{equation}
referred to as the \emph{training set}. We assume the $m$ POVM elements in the training set are drawn from an unknown probability distribution $\mathcal{D}$ over two-outcome POVMs. The objective is to predict the expectation value, within a specified accuracy and failure probability, of a new POVM element $E^{(1)}$ drawn from $\mathcal{D}$ based on the information contained in the training set $T$.
Aaronson \cite{Aaronson_2007} has demonstrated that if the training set has size $m=O(n)$, then an algorithm exists to solve the problem of PAC learning of quantum states. This algorithm is remarkably straightforward: It involves generating a \emph{hypothesis} quantum state $\sigma$, which aligns with all the expectation values in the training set within a given margin of error for each value. More precisely, the theorem can be formulated as follows (using the notation from 
Ref.\ \cite{yoganathan2019condition}):

\begin{theorem}[Quantum Occam's razor \cite{Aaronson_2007}] 
Let $\mathcal{C}_n$ be a class of $n$-qubit quantum states, and $\mathcal{C}=\bigcup_n \mathcal{C}_n$. Let $\mathcal{M}_n$ be a set of two-outcome POVMs on $n$-qubits, and $\mathcal{M}=\bigcup_n \mathcal{M}_n$. Let $\left\{E_i\right\}_{i\in[m]}$ be $m$ POVM elements drawn from the distribution $\mathcal{D}_n$ over $\mathcal{M}_n$, let $\rho \in \mathcal{C}_n$ be an $n$-qubit quantum state, and define the set $T=\left\{\left(E_i, \operatorname{Tr}\left(E_i \rho\right)\right)\right\}_{i \in[m]}$ as the \emph{training set}. Also, fix error parameters $\varepsilon, \eta, \gamma, \delta>0$ with $\gamma \varepsilon \geq 7 \eta$. Suppose there exists a hypothesis $n$-qubit quantum state $\sigma$ such that $$\left|\operatorname{Tr}\left(E_i \sigma\right)-\operatorname{Tr}\left(E_i \rho\right)\right| \leq  \eta $$ for each $i \in[m]$. Then, with probability at least $1-\delta$, we have
\begin{equation}
\operatorname{Pr}_{E \sim D_n}[|\operatorname{Tr}(E \sigma)-\operatorname{Tr}(E \rho)| \leq \varepsilon] \geq 1-\delta
\end{equation}
provided there were enough POVM elements in the training set
\begin{equation}
m \geq \frac{K}{\sigma^2 \varepsilon^2}\left(\frac{n}{\gamma^2 \varepsilon^2} \log ^2 \frac{1}{\gamma \varepsilon}+\log \frac{1}{\delta}\right),
\end{equation}
where $K>0$ is a constant.
\end{theorem}
Roughly speaking, the theorem asserts that if a hypothesis state $\sigma$ is found that accurately aligns with the expectation values in the training set, and if the size of the training set is at least $m=O(n)$, then it is highly probable that the hypothesis state $\sigma$ will maintain its consistency even when confronted with previously unseen POVM elements.
However, it is important to note that finding such a quantum state $\sigma$ that matches the training set expectation values can be formulated as a SDP, which generally has an inefficient computational runtime of $O(\mathrm{poly}(2^n))$. Additionally, even if we manage to find such a state $\sigma$, it may require exponential classical memory to store it. Moreover, making predictions on this state, i.e.,  computing expectation values for unseen POVM elements, can also pose computational challenges.
To address these concerns, Yoganathan \cite{yoganathan2019condition} introduced the concept of time-efficient PAC learning for a given class of quantum states $\mathcal{C}$ under a set of measurements $\mathcal{M}$, which can be stated as follows:

\begin{definition} 
[Efficiently PAC-learnable class \cite{yoganathan2019condition}]
\label{def:EffPAClearn}
Let $\mathcal{C}$ and $\mathcal{M}$ be as described previously, and assume that there exists an efficient classical representation for every $E\in\mathcal{M}_n$, allowing $E$ to be described using $O(\mathrm{poly}(n))$ classical bits. We say that $\mathcal{C}$ is efficiently learnable with respect to $\mathcal{M}$ if there exists a pair of classical algorithms, $L_1$ and $L_2$, that satisfy the following properties:
\begin{itemize}
\item \textbf{($L_1$, SDP feasibility algorithm):} For any $\rho \in \mathcal{C}_n$, $\eta>0$, and any training set $T=\{(E_i,\Tr(E_i \rho))\}_{i\in[m]}$, $L_1$ outputs an $O(\mathrm{poly}(n))$ bit classical description of a hypothesis state $\sigma$ that satisfies the following approximate feasibility problem:

\begin{gather}\label{eq:learningSDP}
\nonumber |\Tr(E_i  \sigma) - \Tr(E_i \rho) | \leq \eta \quad \text{for all} \enspace i \in [m], \\
\sigma \succeq 0, \\ 
\nonumber \Tr(\sigma) = 1.  
\end{gather}

$L_1$ runs in time $O(\mathrm{poly}(n,m,1/\eta))$.

\item \textbf{($L_2$ simulation algorithm):} For every $n$-qubit state $\sigma$ that is an output of $L_1$, and for every $E\in \mathcal{M}_n$, $L_2$ computes $\Tr(E \sigma)$. $L_2$ runs in time $O(\mathrm{poly}(n))$.
\end{itemize}
\end{definition}

This definition ensures that a class of states $\mathcal{C}$ with respect to the set of measurements $\mathcal{M}$ is considered time-efficiently PAC-learnable if it allows for the efficient derivation of a classical description of a quantum state $\sigma$ that aligns with the expectation values in the training set. Furthermore, the classical description of the quantum state $\sigma$ should enable the efficient computation of expectation values for previously unseen POVM elements.
As mentioned, it has been shown that the class of stabilizer states~\cite{rocchetto2018stabiliser} is efficiently PAC learnable in the context of Pauli measurements, and in 
Ref.\ \cite{yoganathan2019condition}, it has been established that the class of low-entangled states can be efficiently PAC-learned under low Schmidt rank measurements~\cite{yoganathan2019condition}. Here, we address the question of whether fermionic Gaussian states, which are also efficiently classically simulable, can be efficiently PAC-learned. Our work demonstrates that ``proper"-PAC learning of fermionic Gaussian states is \texttt{NP}-hard under Majorana or Pauli measurements. Here, ``proper" implies that the hypothesis state $\sigma$ must belong to the same class as the state $\rho$ from which the training set is originated.

To formalize proper-PAC learning, we define it as follows (noting that the only distinction from the previous definition~\ref{def:EffPAClearn} is the requirement that the hypothesis state $\sigma$ belongs to the same class as $\rho$):

\begin{definition} 
\textit{(Proper-efficiently PAC-learnable class)} 
Let $\mathcal{C}_n$ represent a class of $n$-qubit quantum states (e.g., fermionic Gaussian states), with $\mathcal{C}=\bigcup_n \mathcal{C}_n$. Let $\mathcal{M}_n$ be a set of two-outcome POVMs on $n$-qubits, and $\mathcal{M}=\bigcup_n \mathcal{M}_n$ (e.g., POVMs associated with Majorana or Pauli operators). Assuming there exists an efficient classical representation for every $E\in\mathcal{M}_n$, allowing $E$ to be described using $O(\mathrm{poly}(n))$ classical bits, we declare that $\mathcal{C}$ is efficiently proper-PAC-learnable with respect to $\mathcal{M}$ if a pair of classical algorithms $L_1$ and $L_2$ satisfy the following properties:

\begin{itemize}
\item \textbf{($L_1$ proper SDP feasibility algorithm):}  For any $\rho \in \mathcal{C}_n$, $\eta>0$, and any training set $T=\{(E_i,\Tr(E_i \rho))\}_{i\in[m]}$, $L_1$ outputs a $O(\mathrm{poly}(n))$ bit classical description of a hypothesis state $\sigma \in \mathcal{C}_n$ (which is the same class of states from which $\rho$ comes from) that satisfies the following approximate feasibility problem:

\begin{gather}\label{eq:learningSDP}
\nonumber |\Tr(E_i  \sigma) - \Tr(E_i \rho) | \leq \eta \quad \text{for all} \enspace i \in [m], \\
\sigma \succeq 0, \\ 
\nonumber \Tr(\sigma) = 1.  
\end{gather}

$L_1$ runs in time $O(\mathrm{poly}(n,m,1/\eta))$.

\item \textbf{($L_2$ simulation algorithm):} For every $n$-qubit state $\sigma \in \mathcal{C}_n$ output of $L_1$ and for every $E\in \mathcal{M}_n$, $L_2$ computes $\Tr(E \sigma)$. $L_2$ runs in time $O(\mathrm{poly}(n))$.
\end{itemize}
\end{definition}
In a related context, Liang has demonstrated a hardness proof~\cite{Liang_2023} using the concept of proper PAC-learning, specifically in the context of PAC-learning of unitaries \cite{Caro_2020}. His work showed that properly PAC-learning Clifford circuits under Pauli measurements is not efficient, relying on a widely believed assumption in complexity theory.

\section{Proof of the main results}\label{sec:proofs}

We are provided with a training set of Majorana product expectation values estimates (or Pauli expectation values estimates, due to Jordan-Wigner) that arise from an unknown quantum state $\rho$. The question is whether there exists a fermionic Gaussian state $\sigma \in \mathcal{S}_{\mathrm{FG}}$ capable of accurately matching this data. Here, $\mathcal{S}_{\mathrm{FG}}$ denotes the set of fermionic Gaussian states.
Let $\mathcal{S}$ a set of ordered indices $\mathcal{S}\ni S\subseteq [2n]$ encoding Majorana product operators such that their expectation values are given in the training set.  We formalize this problem as follows.

\begin{problem}[$\class{Gaussian\,Consistency}$ problem]\label{prob:pac_approx_dec} 
Problem:
    \begin{itemize}
        \item \textbf{Input:} A training set
            \(\mathcal{T} = \{(S, \left<\gamma_{S}\right>_T)\}_{S\in\mathcal{S}}\) and accuracy parameters $\varepsilon_1,\varepsilon_2 >0$.
            \item \textbf{Promise:} {The training set is obtained from a physical state $\rho$, where $\left<\gamma_{S}\right>_{T} = \langle\gamma_S\rangle_{\rho}$.}
        
        \item \textbf{Output:} 
            \begin{itemize}
                \item \textbf{YES:} There exists a state \(\sigma \in S_{\mathrm{FG}}\) such that 
                \begin{equation}
                    \forall S \in \mathcal{S}: \quad |\left<\gamma_{S}\right>_T- \langle\gamma_S\rangle_{\rho}| \le \varepsilon_1.
                \end{equation}
                
                \item \textbf{NO:} For all states \(\sigma \in S_{\mathrm{FG}}\), 
                \begin{equation}
                    \exists S \in \mathcal{S}: \quad |\left<\gamma_{S}\right>_T- \langle\gamma_S\rangle_{\rho}| > \varepsilon_2.
                \end{equation}
            \end{itemize}
    \end{itemize}
\end{problem}
The inclusion of the problem in \class{NP} is straightforward for all $\varepsilon_1,\varepsilon_2\geq 0$: the correlation matrix serves as the proof, and the expectation values can be calculated efficiently using Wick's theorem. Thus, we will omit this step in the following completeness proofs.
Notice that Problem~\ref{prob:GaussianConsMAIN} is recovered for $\varepsilon=\varepsilon_1=\varepsilon_2$. The complexity of the problem depends on the values $\varepsilon_1$ (completeness) and $\varepsilon_2$ (soundness). In practise, $\left<\gamma_{S}\right>$ may originate from some experiment and is, therefore,  always subject to shot-noise errors which can be accounted for here. If we set $\varepsilon_1=0$, this corresponds to the exact case. Generally, this makes the problem easier and, therefore,  showing hardness with $\varepsilon_1=0$ shows that the complexity is independent of the shot-noise. 
The difference $\varepsilon_2-\varepsilon_1$ shows the accuracy of the algorithm. Hardness, therefore, implies that (assuming $\class{P}\neq \class{NP}$), no polynomial time algorithm can PAC-learn all states up to $\varepsilon_2-\varepsilon_1$ precision.

Now, we mention the high-level idea of our proofs. We know that any fermionic Gaussian state is uniquely determined by its correlation matrix, which is an anti-symmetric real matrix with eigenvalues in absolute value less than one. The correlation matrix is filled with the $2$-point correlation functions, while each higher-order correlation function provides a Pfaffian constraint on the submatrices of the correlation matrix by Wick's theorem. Therefore, our problem of finding a fermionic Gaussian state consistent with the data is equivalent to solving a matrix completion problem (for an anti-symmetric real matrix) with Pfaffian constraints on the submatrices and (linear) constraints on the eigenvalues. Thus, the proof strategy is the one of encoding $3$-SAT istances into such non-linear optimization problem.

It is important to stress that the hardness of the various precisely formulated problems above stems from the fact that non-linear constraints are being imposed in terms of higher order correlations. When imposing constraints arising from mere 2-point correlation functions, one does 
not encounter computationally difficult problems.
Specifically, e.g.,
if one has an index set
$I$ given and matrix entries
$\Xi_{j,k}$, then for every
$\varepsilon>0$, 
one can in polynomial
time find a fermionic Gaussian state $\sigma$
that fulfills
\begin{equation}
|\langle \gamma_j \gamma_k \rangle_\sigma - \Xi_{j,k}|\leq \varepsilon
\end{equation}
for all $(j,k)\in I$, if it exists. This amounts to a simple efficient semi-definite feasibility problem that provides a fermionic Gaussian state $\sigma$, characterized by its
real and anti-symmetric correlation matrix $\Gamma(\sigma)$. This is because the constraint that a skew-symmetric matrix has eigenvalues smaller than or equal to unity can be readily written as a semi-definite constraint as  
\begin{equation}
\label{eq:semidefconst}
I \succeq i \Gamma(\sigma)
\end{equation}
In other words, unsurprisingly, one can computationally efficiently see whether a given Gaussian state is compatible with a collection of entries of the correlation matrix, in a reading of a \emph{Gaussian matrix completion problem}.
\subsection{\texttt{NP}-completeness of \(\class{Gaussian\,Consistency}\) (6-body correlators)}

To establish the \class{NP}-completeness, we endeavor to reduce it to the well-known $3$-SAT problem (Problem~\ref{prob:3sat}). Our objective is to transform instances of the $3$-SAT problem into a format compatible with the input requirements of the $\class{Gaussian\,Consistency}$ problem, allowing us to infer the corresponding output for the $3$-SAT problem from the output of the Gaussian Consistency problem.
To begin, we demonstrate \(\class{NP}\)-completeness in scenarios where the training set comprises solely \(k\)-point correlation functions (with \(k\leq 6\)), and subsequently extend this to training sets which consist solely of $k$-point correlation functions (with $k\leq 4$).
\begin{theorem}[\texttt{NP}-completeness of \(\class{Gaussian\,Consistency}\) (6-body correlators)]
    The \(\class{Gaussian\,Consistency}\) problem is \(\class{NP}\)-complete when the training set consists of \(k\)-point correlation functions with \(k \leq 6\) and accuracy error $\varepsilon_2\le 1.6\%$ independent of the number of modes.
\end{theorem}

\begin{proof}

We aim to demonstrate a polynomial-time reduction from the 3-SAT problem to the \(\class{Gaussian\, Consistency}\) problem for up to 6-point correlations, thereby establishing NP-completeness.
Thus, we construct a training set for the input of the \(\class{Gaussian\, Consistency}\) problem that embeds the \(3\)-SAT clauses.
For this, we first fix some constraints on our correlation matrix. We consider a system of \(2\nx\) fermionic modes to which is associated a \(4\nx \times 4 \nx\) correlation matrix \(\Gamma\), where $n_x$ is the number of $3$-SAT variables. We begin by imposing a particular structure for the \(j\)-th \(4 \times 4\) diagonal block, denoted \(\Gamma_{j,\cdot;j,\cdot} \equiv {\Gamma\rvert_{\{j\}}}\), within the correlation matrix \(\Gamma\) for \(j \in [\nx]\), as 
\begin{align}
    \Gamma_{j,\cdot;j,\cdot} = \frac{1}{\Phi}\begin{pmatrix}
    0 & x_j^0 & 0 & x_j^1 \\
    -x_j^0 & 0 & 1 & 0 \\
    0 & -1 & 0 & 1 \\
    -x_j^1 & 0 & -1 & 0\\
    \end{pmatrix},
    \label{eq:1}
\end{align}
where \(\Phi\) is a constant to be determined later, and \(x_j^0, x_j^1 \in [0,1]\) are free parameters related to the \(3\)-SAT variables. 
We can enforce \(\Gamma_{j,\cdot;j,\cdot}\) to adhere to the above form by introducing the 
constraints
\begin{align}
    \left< \gamma_{j,a}\gamma_{j,b}\right>_T = \Gamma_{j,a;j,b},\label{eq:block_constraints}
\end{align} 
for all \(a<b \in [4]\) except for \((a,b)=(1,2),(1,4)\). Now we adopt the following notation $\gamma_{j,k}:=\gamma_{4(j-1)+k}$, for all $j\in [\nx]$, $k\in [4]$. 
Additionally, we demand that
\begin{align}
\left<\gamma_{j,1}\gamma_{j,2}\gamma_{j,3}\gamma_{j,4}\right>_T=\frac{1}{\Phi^2},
\end{align}
for all $j\in[\nx]$.
By Wick's theorem, we have
\begin{align}
\left<\gamma_{j}\gamma_{j+1}\gamma_{j+2}\gamma_{j+3}\right> = \Pf({\Gamma\rvert_{\{j\}}}) =(\Gamma\rvert_{\{j\}})_{1,2}(\Gamma\rvert_{\{j\}})_{3,4}+(\Gamma\rvert_{\{j\}})_{2,3}(\Gamma\rvert_{\{j\}})_{1,4}=\frac{1}{\Phi^2}(x^0_j+x^1_j)\,.
\end{align}
Hence, we get that $x_j^1=1-x_j^0$. 
We require the correlation matrix $\Gamma$ to have the block diagonal form
    \begin{align}
        \Gamma=\Gamma\rvert_{\{1\}}\oplus \cdots\oplus \Gamma\rvert_{\{\nx\}},
    \end{align}
by demanding all the elements of $\Gamma$ outside the diagonal blocks to be zero, which means that in the training set all the values outside the block-diagonal are demanded to be zero. In formulae, this is 
\begin{align}
    \left< \gamma_{j,a}\gamma_{k,b}\right>_T=0 \quad \forall j\neq k\in[\nx];\, a,b\in [4]\,.
\end{align}
We are now prepared to encode the $n_c$ $3$-SAT clauses. We have
\begin{align}
    \forall c \in [n_c]: \quad s_{1,c} x_{i_1(c)} \lor s_{2,c} x_{i_2(c)} \lor s_{3,c} x_{i_3(c)},
\end{align}
where \(\vec i \in [\nx]^{3,n_c}\) represents a list of indices $i_{j}(c)$ for $c\in[n_c]$ and $j\in[3]$, and \(s \in \{\cdot, \lnot\}^{3,n_c}\) is a negation table. Here, $[\nx]^{3,n_c}$ means that we take $n_c$ subsets of $3$ elements belonging to $[\nx]$. Similarly for \(s \in \{\cdot, \lnot\}^{3,n_c}\).
For \(c \in [n_c]\), we define \(s_j(c) \in \{0,1\}\) such that \(s_j(c) = 1\) if \(s_{j,c} = \lnot\), and \(s_j(c) = 0\) otherwise for $j\in[3]$. With the above definitions, we thus have the following notation for Majorana operators $\gamma_{i_j(c),k}=\gamma_{4(i_j(c)-1)+k}$ for all $j\in [\nx]$ and $k\in[4]$. For each clause, we require the following six-point correlators to be zero as
\begin{align}
 C_{\vec i(c),\vec s(c)}\coloneqq\left<\gamma_{i_1(c),1}\gamma_{i_1(c),2+2s_1(c)}\gamma_{i_2(c),1}\gamma_{i_2(c),2+2s_2(c)}\gamma_{i_3(c),1}\gamma_{i_3(c),2+2s_3(c)}\right> = 0,
    \label{eq:6corr2}
\end{align}
for all \(c \in [n_c]\). By Wick's theorem, these correlators correspond to the Pfaffian of a block diagonal $6 \times 6 $ sub-matrix, as 
\begin{align}
    C_{\vec i(c),\vec s(c)}&\coloneqq \operatorname{Pf}\left(\frac{1}{\Phi}\bigoplus_{j = 1}^{3} \begin{pmatrix} 0 &  x^{s_j(c)}_{i_j(c)} \\ -x^{s_j(c)}_{i_j(c)} & 0 \end{pmatrix}\right)  \\
    &= \frac{1}{\Phi^3}\prod^3_{j=1} x_{i_j(c)}^{s_j(c)},
    \nonumber
\end{align}
where we have defined ${\vec i}(c)\coloneqq\{i_1(c), i_2(c), i_3(c)\}$ and ${\vec s}(c)\coloneqq \{s_1(c), s_2(c), s_3(c)\}$.
Hence, with these imposed six-point correlators, we encode the conditions
\begin{align}
    x_{i_1(c)}^{s_1(c)}x_{i_2(c)}^{s_2(c)}x_{i_3(c)}^{s_3(c)} = 0,
    \label{eq:constr}
\end{align}
for all \(c \in [n_c]\). 
Additionally, we set $\Phi = \frac{1 + \sqrt{5}}{2}$ to constrain the eigenvalues of $\Gamma$, ensuring it is a valid correlation matrix. { The chosen $\Phi$ corresponds to the smallest prefactor that still allows for the correlation to be of a valid quantum state. Recall that for a correlation matrix to be valid, its singular values must be less than or equal to one (see Eq.~\ref{eq:semidefconst}).}

\paragraph{Completeness:} To show completeness, we need to show that when the $3$-SAT instance is satisfiable, then we can also construct a Gaussian state which satisifies all the expectation values. For this we define $(x_j^0,x_j^1)=(0,1)$, if $x_j=\mathrm{True}$ and $(x_j^0,x_j^1)=(1,0)$, if $x_j=\mathrm{False}$. By direct calculation, we can verify that $\Gamma\rvert_{\{j\}}$ describe valid correlation matrices and, therefore, so does the entire $\Gamma$. Also it follows that 
\begin{align}
    s_{1,c} x_{i_1(c)}\lor s_{2,c} x_{i_2(c)}\lor s_{3,c} x_{i_3(c)}=\mathrm{True}\Longleftrightarrow x_{i_1(c)}^{s_1(c)}x_{i_2(c)}^{s_2(c)}x_{i_3(c)}^{s_3(c)} = 0\,
\end{align}
which completes the completeness analysis.
\paragraph{Soundness:} Here we need to show that if the $3$-SAT instance does not have a solution, then there cannot exist a correlation matrix which satisifies all expectation values up to $\varepsilon_2$ precision. To see this we show that for sufficiently small $\varepsilon_2$, a valid correlation matrix would imply also a solution to the $3$-SAT instance. For this, we define the Boolean variables
\begin{align}
\label{eq:rules}
    x_i=\begin{cases}
        \mathrm{True} & x^1_j\geq x^0_j,\\
        \mathrm{False} & x^1_j< x^0_j,
    \end{cases}
\end{align}
where we have defined $x_j^s\coloneqq\Phi \Gamma_{j,1;j,2+2s}$ for $j\in[\nx]$ and $s\in \{0,1\}$.
In the exact case we have $x_j^0+x_j^1=1$. Due to tolerances in the expectation values, this now only holds approximately
\begin{equation}
        |x_j^0+x_j^1-1|\leq \varepsilon_x\,.
\end{equation}
Now we impose 
\begin{align}
    \left<\gamma_{j,1}\gamma_{j,2}\gamma_{j,3}\gamma_{j,4}\right> &= \mathrm{Pf}(\Gamma\rvert_{\{j\}})
    = \frac{1}{\Phi}(x_j^0 \Gamma_{j,3;j,4}+  x_j^1 \Gamma_{j,2;j,3})-\Gamma_{j,1;j,3}\Gamma_{j,2;j,4},
\end{align}   
as well as all the relevant measurement constraints
\begin{align}
    \left|\left<\gamma_{j,1}\gamma_{j,2}\gamma_{j,3}\gamma_{j,4}\right>-\frac{1}{\Phi^2}\right|,\,
    \left|\Gamma_{j,3;j,4}-\frac{1}{\Phi}\right|,\,
    \left|\Gamma_{j,2;j,3}-\frac{1}{\Phi}\right|,\,
    \left|\Gamma_{j,1;j,3}\right|,\,
    \left|\Gamma_{j,2;j,4}\right|&\leq \varepsilon_2\,.
\end{align}
Our objective is to find a suitable $\varepsilon_x$. Using triangle inequalities, we get
\begin{align}
    \left|\frac{1}{\Phi}(x_j^0 \Gamma_{j,3;j,4}+  x_j^1 \Gamma_{j,2;j,3})-\Gamma_{j,1;j,3}\Gamma_{j,2;j,4}-\frac{1}{\Phi^2}\right|&\leq \varepsilon_2,\\
    \left|\frac{1}{\Phi}(x_j^0 \Gamma_{j,3;j,4}+  x_j^1 \Gamma_{j,2;j,3})-\frac{1}{\Phi^2}\right|-\left|\Gamma_{j,1;j,3}\Gamma_{j,2;j,4}\right|&\leq \varepsilon_2,\\
    \left|\frac{1}{\Phi}(x_j^0 \Gamma_{j,3;j,4}+  x_j^1 \Gamma_{j,2;j,3})-\frac{1}{\Phi^2}\right|-\varepsilon_2^2&\leq \varepsilon_2,\\
    \left|\frac{1}{\Phi^2}(x_j^0+  x_j^1)-\frac{1}{\Phi^2}+\frac{x_j^0}{\Phi} (\Gamma_{j,3;j,4}-\frac{1}{\Phi})+ \frac{x_j^1}{\Phi}  (\Gamma_{j,2;j,3}-\frac{1}{\Phi})\right|&\leq \varepsilon_2+\varepsilon_2^2,\\
    \left|\frac{1}{\Phi^2}(x_j^0+  x_j^1)-\frac{1}{\Phi^2}\right|-\left|\frac{x_j^0}{\Phi} (\Gamma_{j,3;j,4}-\frac{1}{\Phi})\right|-\left|\frac{x_j^1}{\Phi}  (\Gamma_{j,2;j,3}-\frac{1}{\Phi})\right|&\leq \varepsilon_2+\varepsilon_2^2,\\
    \frac{1}{\Phi^2}|x_j^0+x_j^1-1|&\leq 3\varepsilon_2+\varepsilon_2^2.
\end{align}
Here, we have used that $x_j^\alpha/\Phi\leq 1$, since it is an entry of the correlation matrix. Therefore,  we can choose $\varepsilon_x=(3\varepsilon_2+\varepsilon_2^2)\Phi^2$.
Similarly, for the clause we also get an error
\begin{align}
    \left| {x}^{s_1(c)}_{i_1(c)}{x}^{s_2(c)}_{i_2(c)}{x}^{s_3(c)}_{i_3(c)}\right|\leq \varepsilon_l,
        \label{eq:3satnoisy}
\end{align}
which we can derive by using
\begin{align}
C_{\vec{i}(c),\vec{s}(c)}=\Pf(\Gamma_{{\vec i}(c),{\vec s}(c)})\quad \textrm{and}\quad 
    \left|C_{\vec{i}(c),\vec{s}(c)}\right|\leq \varepsilon_2,
\end{align}
where we recall that
\begin{align}
     C_{\vec i(c),\vec s(c)}\coloneqq\left<\gamma_{i_1(c),1}\gamma_{i_1(c),2+2s_1(c)}\gamma_{i_2(c),1}\gamma_{i_2(c),2+2s_2(c)}\gamma_{i_3(c),1}\gamma_{i_3(c),2+2s_3(c)}\right>.
\end{align}
Now our objective is to find $\varepsilon_{l}$. The Pfaffian of a $6 \times 6$ matrix contains $15$ terms. Besides the values $\Gamma_{{i}_{k}(c),2+2s_{k}(c)}$ for $ k\in [3]$, all values are at most $\varepsilon_2$ in size (by the expectation value constraints). Therefore,  except for the term $\prod_{k\in [3]}\Gamma_{{i}_{k}(c),2+2s_{k}(c)}=\frac{1}{\Phi^3}{x}^{s_1(c)}_{i_1(c)}{x}^{s_2(c)}_{i_2(c)}{x}^{s_3(c)}_{i_3(c)}$, all other terms are at most $\varepsilon_2^2$ in magnitude.
Thus, we have
\begin{align}
    \left|\frac{{x}^{s_1(c)}_{i_1(c)}{x}^{s_2(c)}_{i_2(c)}{x}^{s_3(c)}_{i_3(c)}}{\Phi^3}\right|&\leq |\Pf(\Gamma_{{\vec i}(c),{\vec s}(c)})|+\left|\Pf(\Gamma_{{\vec i}(c),{\vec s}(c)})-\frac{{x}^{s_1(c)}_{i_1(c)}{x}^{s_2(c)}_{i_2(c)}{x}^{s_3(c)}_{i_3(c)}}{\Phi^3}\right|\\
    &\leq \varepsilon_2+ 14\varepsilon_2^2
    \nonumber
\end{align}
which gives $\varepsilon_l=(\varepsilon_2+ 14\varepsilon_2^2)\Phi^3$.
Lastly, we need that 

\begin{align}
    s_{1,c} x_{i_1(c)}\lor s_{2,c} x_{i_2(c)}\lor s_{3,c} x_{i_3(c)}=\mathrm{False} \Rightarrow 
    \left| {x}^{s_1(c)}_{i_1(c)}{x}^{s_2(c)}_{i_2(c)}{x}^{s_3(c)}_{i_3(c)}\right|\geq \varepsilon_l\,,
\end{align}
since this implies that there cannot exist a valid state corresponding to a false proof string for the $3$-SAT instance. We can use that 
\begin{align}
    \mathrm{max}(x_j^0,x_j^0)\geq \frac{1-\varepsilon_x}{2},
\end{align}
which follows from ${x}^{0}_{j}+{x}^{1}_{j}\geq 1- \varepsilon_x$.
Since this needs to hold for all three variables in the clause, we obtain the expression
\begin{align}
    \left(\frac{1-\varepsilon_x}{2}\right)^3\geq \varepsilon_l
\end{align}
which holds for $\varepsilon_2\leq 1.6\%$.

\end{proof}

\subsection{\texttt{NP}-Completeness of \(\class{Gaussian\,Consistency}\) (4-body correlators)}

\begin{theorem}[\texttt{NP}-completeness of \(\class{Gaussian\,Consistency}\) (4-body correlators)]
The $\class{Gaussian\, Consistency}$ problem is $\class{NP}$-complete when the training set consists of $k$-point correlation functions with $k \leq 4$ and accuracy error $\varepsilon_2 \le 0.069\% $ independent from the number of modes.
\end{theorem}

\begin{proof}
To reduce the weight of the Majorana operators to 4, we modify the correlation matrix structure from the previous proof using Eq.~\eqref{eq:block_constraints}. We keep the same diagonal blocks. We recall that $\gamma_{i_j(c),k}=\gamma_{4(i_j(c)-1)+k}$ for all $j\in [\nx]$ and $k\in[4]$.
Here, \( \Phi \) is defined as a constant that will be fixed later in the proof.

Instead of employing the 6-point correlators (as in Eq.~\eqref{eq:6corr2}), for each clause $c \in [n_c]$, we impose the following conditions 
\begin{align}
\begin{aligned}
&\left<\gamma_{i_1(c),1}\gamma_{i_2(c),1}\right>_T = \frac{1}{\Phi}, \\
&\left<\gamma_{i_1(c),1}\gamma_{i_2(c),2+2s_2(c)}\right>_T = 0, \\
&\left<\gamma_{i_1(c),2+2s_1(c)}\gamma_{i_2(c),1}\right>_T = 0, \\
&\left<\gamma_{i_1(c),1}\gamma_{i_1(c),2+2s_1(c)}\gamma_{i_2(c),1}\gamma_{i_2(c),2+2s_2(c)}\right>_T = 0,\\
&\left<\gamma_{i_1(c),2+2s_1(c)}\gamma_{i_3(c),1}\right>_T = 0, \\
&\left<\gamma_{i_1(c),2+2s_1(c)}\gamma_{i_3(c),1}\right>_T = 0,\\
&\left<\gamma_{i_1(c),2+2s_1(c)}\gamma_{i_2(c),2+2s_2(c)}\gamma_{i_3(c),1}\gamma_{i_3(c),2+2s_3(c)}\right>_T = 0.
\end{aligned}
\end{align}
Let us analyze first the $\varepsilon_2=0$ case. The first four constrains lead to
 \begin{align*}
0=\langle\gamma_{i_{1}(c),1}\gamma_{i_{1}(c),2+2s_{1}(c)}\gamma_{i_{2}(c),1}\gamma_{i_2(c),2+2s_2(c)}\rangle_\sigma &= \frac{1}{\Phi^2}\mathrm{Pf}\begin{pmatrix}
0 & x_{i_1(c)}^{s_1(c)} & 1 & 0\\
-x_{i_1(c)}^{s_1(c)} & 0 & 0 & y^{s_1(c)s_2(c)}_{i_1(c)i_2(c)}\\
-1 & 0 & 0 & x_{i_2(c)}^{s_2(c)}\\
0 & -y^{s_1(c)s_2(c)}_{i_1(c)i_2(c)} & -x_{i_2(c)}^{s_2(c)} & 0
\end{pmatrix}\\
&= (x_{i_1(c)}^{s_1(c)}x_{i_2(c)}^{s_2(c)} - y^{s_1(c)s_2(c)}_{i_1(c)i_2(c)})\frac{1}{\Phi^2},
\end{align*}
where $y^{s_1(c)s_2(c)}_{i_1(c)i_2(c)}=\Phi\Gamma_{i_1(c),s_1(c);i_2(c),s_2(c)}$ is an element of the matrix to be fixed.  The last three constraints give
\begin{align*}
\left<\gamma_{i_1(c),2+2s_1(c)}\gamma_{i_2(c),2+2s_2(c)}\gamma_{i_3(c),1}\gamma_{i_3(c),2+2s_3(c)}\right> &= \frac{1}{\Phi^2}\mathrm{Pf}\begin{pmatrix}
0 & y^{s_1(c)s_2(c)}_{i_1(c)i_2(c)} & 0 & *\\
-y^{s_1(c)s_2(c)}_{i_1(c)i_2(c)} & 0 & 0 & *\\
0 & 0 & 0 & x_{i_3(c)}^{s_3(c)}\\
* & * & -x_{i_3(c)}^{s_3(c)}& 0
\end{pmatrix}\\
&= \frac{1}{\Phi^2}y^{s_1(c)s_2(c)}_{i_1(c)i_2(c)}x_{i_3(c)}^{s_3(c)}, 
\end{align*}
where the value of $(*)$ may be zero or some $y^{s_1(c^{\prime})s_2(c^{\prime})}_{i_1(c^{\prime})i_2(c^{\prime})}$ for some clause $c^{\prime}\neq c$.

\paragraph{Completeness:} For completeness, we use the same encoding as in the previous proof. We choose $\Gamma_{i_1(c),s_1(c);i_2(c),s_2(c)}=\frac{1}{\Phi}x_{i_1(c)}^{s_1(c)}x_{i_2(c)}^{s_2(c)}$ and any undefined entry in $\Gamma$ is set to $0$. By construction, this reproduces all desired expectation values. For $\Gamma$ to be a valid correlation matrix, we also require that its operator norm if bounded by $1$.
As such, we can decompose the matrix into its (block)-diagonal part 
and their off-diagonal component
\begin{align}
    \Gamma=\Gamma_{\mathrm{diag}}+\Gamma_{\mathrm{off}}.
\end{align}
From the previous analysis, we know that 
\begin{equation}
\|\Gamma_{\mathrm{diag}}\|_\infty=\frac{1+\sqrt{5}}{2\Phi}.
\end{equation}
To bound the operator norm of the off diagonal matrix, we use that the authors of Ref.\ \cite{TOVEY198485} have shown that 3-SAT is \texttt{NP}-complete, even if every Boolean variable is present only in at most 4 clauses. This means that each column only has four non vanishing entries.

{We now first recall the Gershgorin Circle Theorem: Let \( A \in \mathbb{C}^{m \times m} \) be a square matrix, and define the Gershgorin discs as $
D_i \coloneqq \{ z \in \mathbb{C} : |z - A_{ii}| \leq \sum_{j \neq i} |A_{ij}| \},$ for each $i \in [m]$.
Then, all eigenvalues of \( A \) lie within the union of these discs $\lambda(A) \subseteq \bigcup_{i=1}^m D_i.$}

{Using this theorem, and recalling the operator norm is equal to the largest eigenvalue, we find:} 
\begin{align}
\|\Gamma_{\mathrm{off}}\|_\infty \leq \max_{i\in [4 \nx]}\sum^{4 \nx}_{j=1}|(\Gamma_{\mathrm{off}})_{i,j}| \leq \frac{4}{\Phi}
\end{align}
giving
\begin{align}
    \|\Gamma\|_\infty\leq\|\Gamma_{\mathrm{diag}}\|_\infty+ \|\Gamma_{\mathrm{off}}\|_\infty\leq \frac{1}{\Phi}\left(\frac{1+\sqrt{5}}{2}+4\right)\,.
\end{align}
This means that for $\Phi=\frac{9+\sqrt{5}}{2}$, we have a valid correlation matrix.

\paragraph{Soundness:} Here we proceed to perform the error analysis.
For the diagonal components, we have the same condition  $\forall j \in [\nx]$
\begin{align}
        |x_j^0+x_j^1-1|\leq \varepsilon_x=\Phi^2(3\varepsilon_2+\varepsilon_2^2)\,.
\end{align}
Similarly, we get from the first four constraints

\begin{align*}
    \varepsilon_2 &\geq \left| \langle \gamma_{i_1(c),1} \gamma_{i_1(c),2+2s_1(c)} \gamma_{i_2(c),1} \gamma_{i_2(c),2+2s_2(c)} \rangle \right|, \\
    \varepsilon_2 &\geq \left| \frac{1}{\Phi^2} x_{i_1(c)}^{s_1(c)} x_{i_2(c)}^{s_2(c)} - \frac{1}{\Phi} \Gamma_{i_1(c),1;i_2(c),1} y^{s_1(c)s_2(c)}_{i_1(c)i_2(c)} + \Gamma_{i_1(c),1;i_2(c),2+2s_2(c)} \Gamma_{i_1(c),2+2s_1(c);i_2(c),1} \right|, \\
    \varepsilon_2 &\geq \left| \frac{1}{\Phi^2} x_{i_1(c)}^{s_1(c)} x_{i_2(c)}^{s_2(c)} - \frac{1}{\Phi^2} y^{s_1(c)s_2(c)}_{i_1(c)i_2(c)} \right| - \left| \frac{1}{\Phi} y^{s_1(c)s_2(c)}_{i_1(c)i_2(c)} (\Gamma_{i_1(c),1;i_2(c),1} - \frac{1}{\Phi}) \right| \\
    &\phantom{{}=}- \left| \Gamma_{i_1(c),1;i_2(c),2+2s_2(c)} \Gamma_{i_1(c),2+2s_1(c);i_2(c),1} \right|, \\
    \varepsilon_2 &\geq \frac{1}{\Phi^2} \left| x_{i_1(c)}^{s_1(c)} x_{i_2(c)}^{s_2(c)} - y^{s_1(c)s_2(c)}_{i_1(c)i_2(c)} \right| - \varepsilon_2 - \varepsilon_2^2, \\
    2\varepsilon_2 + \varepsilon_2^2 &\geq \frac{1}{\Phi^2} \left| x_{i_1(c)}^{s_1(c)} x_{i_2(c)}^{s_2(c)} - y^{s_1(c)s_2(c)}_{i_1(c)i_2(c)} \right|,
\end{align*}
and from the last constraints
\begin{align*}
    \varepsilon_2 &\geq \left| \left\langle \gamma_{i_1(c),2+2s_1(c)} \gamma_{i_2(c),2+2s_2(c)} \gamma_{i_3(c),1} \gamma_{i_3(c),2+2s_3(c)} \right\rangle \right|, \\
    \varepsilon_2 &\geq \left| \frac{1}{\Phi^2} y^{s_1(c)s_2(c)}_{i_1(c)i_2(c)} x_{i_3(c)}^{s_3(c)} + \Gamma_{i_1(c),2+2s_1(c);i_3(c),2+2s_3(c)} \Gamma_{i_2(c),2+2s_2(c);i_3(c),1} \right. \\
    &\phantom{{}= \left| \frac{1}{\Phi^2} \right.} \left. - \Gamma_{i_1(c),2+2s_1(c);i_3(c),1} \Gamma_{i_2(c),2+2s_2(c);i_3(c),2+2s_3(c)} \right|, \\
    \varepsilon_2 &\geq \frac{1}{\Phi^2} \left| y^{s_1(c)s_2(c)}_{i_1(c)i_2(c)} x_{i_3(c)}^{s_3(c)} \right| - 2\varepsilon_2, \\
   3\varepsilon_2 &\geq \frac{1}{\Phi^2} \left| y^{s_1(c)s_2(c)}_{i_1(c)i_2(c)} x_{i_3(c)}^{s_3(c)} \right|.
\end{align*}
We, therefore,  have
\begin{align}
    \left|x_{i_1(c)}^{s_1(c)}x_{i_2(c)}^{s_2(c)}x_{i_3(c)}^{s_3(c)}\right|&=\left|y^{s_1(c)s_2(c)}_{i_1(c)i_2(c)}x_{i_3(c)}^{s_3(c)}+x_{i_3(c)}^{s_3(c)}(x_{i_1(c)}^{s_1(c)}x_{i_2(c)}^{s_2(c)}-y^{s_1(c)s_2(c)}_{i_1(c)i_2(c)}x_{i_3(c)}^{s_3(c)}) \right|\\
    \nonumber
    &\leq \left|y^{s_1(c)s_2(c)}_{i_1(c)i_2(c)}x_{i_3(c)}^{s_3(c)}\right|+\left|x_{i_3(c)}^{s_3(c)}(x_{i_1(c)}^{s_1(c)}x_{i_2(c)}^{s_2(c)}- y^{s_1(c)s_2(c)}_{i_1(c)i_2(c)}x_{i_3(c)}^{s_3(c)}) \right|\\
    \nonumber
    &\leq \Phi^2(5\varepsilon_2+\varepsilon_2^2)\eqqcolon \varepsilon_l\,.
    \nonumber
\end{align}
For the same reasoning as before, we require
\begin{align}
    \left(\frac{1-\varepsilon_x}{2}\right)^3\geq \varepsilon_l
\end{align}
which holds for $\varepsilon_2\leq 6.5\times 10^{-4}$.
\end{proof}
It is worth noting that, with the help of Lemma~\ref{le:purification}, the previous results easily extend to pure fermionic Gaussian states. 

\subsection{Physicality of the state}\label{sec:physicality}
In this section we show that the task at hand is difficult because the required state needs to be free-fermionic, not because it is hard to find any that fits the statistics in the NO case. For this we can construct a general state which always satisfies the constraints regardless of the 3-SAT instance having a solution or not.
Utilizing convex programming techniques, we derive the state in the Jordan-Wigner basis $\{\ket{0,0},\ket{1,0},\ket{0,1},\ket{1,1}\}$ as 
\begin{align}
    \rho_s =\begin{pmatrix}
          0.4375 &  0.0    &  0.0 &    0.125\\
         0  &   0.3125 & 0.125  & 0\\
        0   & 0.125   & 0.0625  & 0\\
        0.125 &  0    &  0   &   0.1875\\
    \end{pmatrix}
\end{align}
which has the correlation matrix
\begin{align}
    \Gamma(\rho_s) = \frac{1}{2}\begin{pmatrix}
    0 & 0 & 0 & 0 \\
    0 & 0 & 1 & 0 \\
    0 & -1 & 0 & 1 \\
    0 & 0 & -1 & 0\\
    \end{pmatrix},
\end{align}
with 
\begin{align}
\left<\gamma_{1}\gamma_{2}\gamma_{3}\gamma_{4}\right>_{\rho_s}=\frac{1}{4},\quad
    \left<\gamma_j\right>_{\rho_s}=0.
\end{align}
for $j\in [4]$.
Therefore, $\rho=\rho_s^{\otimes \nx}$ satisfies all the constraints from the proof with the six body correlators, as they all vanish. 
In particular, by employing the Jordan-Wigner transformation and leveraging the fact that $\left<\gamma_{\mu}\right>_{\rho_s}=0$ for all $\mu\in[\nx]$, we can demonstrate that for any $i < j \in [\nx]$ and $\alpha, \beta \in [4]$, we have
\begin{align}
    \left<\gamma_{i,\alpha}\gamma_{j,\beta}\right>_{\rho_s^{\otimes \nx}} = 0.
\end{align}
The normalization constant is $\Phi=2$. Therefore,  we also obtain a constant ($\varepsilon_2=0.9\%$) albeit sightly worse error threshold. 

For the four-body correlators, we can equally define a state. We note that the only non trivial constraints are for each clause
 \begin{align}
     \left<\gamma_{i_1(c),1}\gamma_{i_2(c),1}\right>_{T}= \frac{1}{\Phi}.
 \end{align}
 To accommodate this, we define the Gaussian state $\sigma$, for which for each clause
\begin{align}
    \left<\gamma_{i_1(c),1}\gamma_{i_2(c),1}\right>_{\sigma}= \frac{1}{4}
\end{align}
and every other entry vanishes. This is a physical state, since every index is present in at most $4$ clauses by design.
The final constructed state then takes the form
\begin{align}
    \rho=a_1\rho_S^{\otimes \nx}+a_2 \sigma .
\end{align}
Using the clause constraint, we obtain that $\frac{a_2}{4}=\frac{1}{\Phi}$ and, therefore,  $a_1=1-\frac{4}{\Phi}$. This means that $\rho_S$ needs to be modified to have the expectation values
\begin{align}
    \Gamma(\rho_s) = \frac{1}{\Phi a_1}\begin{pmatrix}
    0 & 0 & 0 & 0 \\
    0 & 0 & 1 & 0 \\
    0 & -1 & 0 & 1 \\
    0 & 0 & -1 & 0\\
    \end{pmatrix},
\end{align}
and
\begin{align}
    \left<\gamma_{1}\gamma_{2}\gamma_{3}\gamma_{4}\right>_{\rho_s}=\frac{1}{\Phi^2 a_1}.
\end{align}
Using again convex programming, this has a feasible state for the original normalization constant $\Phi=\frac{9+\sqrt{5}}{2}$
with 
\begin{align}
    \rho_s \simeq\begin{pmatrix}
          0.432 &  0.0    &  0.0 &    0.1545\\
         0  &   0.377 & 0.1545  & 0\\
        0   & 0.1545   & 0.068  & 0\\
        0.1545 &  0    &  0   &   0.123\\
    \end{pmatrix}\,.
\end{align}
This means that correlation matrices constructed in the original proof already correspond to physical states.

\section*{Acknowledgments}
We thank Tommaso Guaita, Alexander Nietner, Daniel Liang and Salvatore F.E. Oliviero for useful discussions. LB has been funded by  DFG (FOR 2724, CRC 183), BMWK (EniQmA), the Cluster of Excellence MATH+ and the BMBF (MuniQC-Atoms). AAM and LL have been funded by the BMBF (FermiQP). JE has additionally been funded by the ERC (DebuQC).

\bibliographystyle{halpha}
\bibliography{ref}

\newcommand{\etalchar}[1]{$^{#1}$}
\begin{thebibliography}{LMMH{\etalchar{+}}17}
\expandafter\ifx\csname url\endcsname\relax
  \def\url#1{\texttt{#1}}\fi
\expandafter\ifx\csname doi\endcsname\relax
  \def\doi#1{\burlalt{doi:#1}{http://dx.doi.org/#1}}\fi
\expandafter\ifx\csname urlprefix\endcsname\relax\def\urlprefix{URL }\fi
\expandafter\ifx\csname href\endcsname\relax
  \def\href#1#2{#2}\fi
\expandafter\ifx\csname burlalt\endcsname\relax
  \def\burlalt#1#2{\href{#2}{#1}}\fi

\bibitem[AA23]{anshu2023survey}
Anurag Anshu and Srinivasan Arunachalam.
\newblock A survey on the complexity of learning quantum states, 2023,
  \burlalt{2305.20069}{http://arxiv.org/abs/2305.20069}.

\bibitem[Aar07]{Aaronson_2007}
Scott Aaronson.
\newblock The learnability of quantum states.
\newblock {\em Proc. Roy. Soc. A}, 463:3089--3114, sep 2007.
\newblock \doi{10.1098/rspa.2007.0113}.

\bibitem[Aar20]{aaronson2018shadow}
Scott Aaronson.
\newblock Shadow tomography of quantum states.
\newblock {\em SIAM Journal on Computing}, 49(5):STOC18--368--STOC18--394,
  2020.
\newblock \doi{10.1137/18M120275X}.

\bibitem[AB09]{arora_barak_2009}
Sanjeev Arora and Boaz Barak.
\newblock {\em Computational complexity: A modern approach}.
\newblock Cambridge University Press, 2009.
\newblock \doi{10.1017/CBO9780511804090}.

\bibitem[ABDY23]{arunachalam2023optimal}
Srinivasan Arunachalam, Sergey Bravyi, Arkopal Dutt, and Theodore~J. Yoder.
\newblock Optimal algorithms for learning quantum phase states.
\newblock 2023, \burlalt{2208.07851}{http://arxiv.org/abs/2208.07851}.

\bibitem[AG23]{aaronson2023efficient}
Scott Aaronson and Sabee Grewal.
\newblock Efficient tomography of non-interacting fermion states, 2023,
  \burlalt{2102.10458}{http://arxiv.org/abs/2102.10458}.

\bibitem[Bax82]{Baxter:1982zz}
Rodney~J. Baxter.
\newblock {\em {Exactly solved models in statistical mechanics}}.
\newblock Academic Press, 1982.
\newblock \doi{10.1142/9789814415255\_0002}.

\bibitem[BBC{\etalchar{+}}19]{Bravyi_2019}
Sergey Bravyi, Dan Browne, Padraic Calpin, Earl Campbell, David Gosset, and
  Mark Howard.
\newblock Simulation of quantum circuits by low-rank stabilizer decompositions.
\newblock {\em Quantum}, 3:181, September 2019.
\newblock \doi{10.22331/q-2019-09-02-181}.

\bibitem[BC23]{begušić2023fast}
Tomislav Begušić and Garnet Kin-Lic Chan.
\newblock Fast classical simulation of evidence for the utility of quantum
  computing before fault tolerance, 2023,
  \burlalt{2306.16372}{http://arxiv.org/abs/2306.16372}.

\bibitem[BEG{\etalchar{+}}24]{Bluvstein_2023}
Dolev Bluvstein, Simon~J. Evered, Alexandra~A. Geim, Sophie~H. Li, Hengyun
  Zhou, Tom Manovitz, Sepehr Ebadi, Madelyn Cain, Marcin Kalinowski, Dominik
  Hangleiter, J.~Pablo Bonilla~Ataides, Nishad Maskara, Iris Cong, Xun Gao,
  Pedro Sales~Rodriguez, Thomas Karolyshyn, Giulia Semeghini, Michael~J.
  Gullans, Markus Greiner, Vladan Vuletić, and Mikhail~D. Lukin.
\newblock Logical quantum processor based on reconfigurable atom arrays.
\newblock {\em Nature}, 626(7997):58–65, 2024.
\newblock \doi{10.1038/s41586-023-06927-3}.

\bibitem[BG16]{Bravyi_2016}
Sergey Bravyi and David Gosset.
\newblock {Improved classical simulation of quantum circuits dominated by
  Clifford gates}.
\newblock {\em Phys. Rev. Lett.}, 116:250501, June 2016.
\newblock \doi{10.1103/physrevlett.116.250501}.

\bibitem[BGC24]{Begu_i__2024}
Tomislav Begušić, Johnnie Gray, and Garget K.-L. Chan.
\newblock Fast and converged classical simulations of evidence for the utility
  of quantum computing before fault tolerance.
\newblock {\em Science Adv.}, 10(3):eadk4321, 2024.
\newblock \doi{10.1126/sciadv.adk4321}.

\bibitem[BMEL24]{Bittel2024testing}
Lennart Bittel, Antonio~Anna Mele, Jens Eisert, and Lorenzo Leone.
\newblock Optimal trace-distance bounds for free-fermionic states: Testing and
  improved tomography, 2024,
  \burlalt{2409.17953}{http://arxiv.org/abs/2409.17953}.
\newblock \urlprefix\url{https://arxiv.org/abs/2409.17953}.

\bibitem[Bra04]{BravyiFermions}
Sergey Bravyi.
\newblock Lagrangian representation for fermionic linear optics.
\newblock 2004,
  \burlalt{quant-ph/0404180}{http://arxiv.org/abs/quant-ph/0404180}.

\bibitem[CD20]{Caro_2020}
Matthias~C. Caro and Ishaun Datta.
\newblock Pseudo-dimension of quantum circuits.
\newblock {\em Quantum Machine Intelligence}, 2:14, nov 2020.
\newblock \doi{10.1007/s42484-020-00027-5}.

\bibitem[CPF{\etalchar{+}}10]{Cramer_2010}
Marcus Cramer, Martin~B. Plenio, Steven~T. Flammia, Rolando Somma, David Gross,
  Stephen~D. Bartlett, Olivier Landon-Cardinal, David Poulin, and Yi-Kai Liu.
\newblock Efficient quantum state tomography.
\newblock {\em Nature Comm.}, 1(1):149, 2010.
\newblock \doi{10.1038/ncomms1147}.

\bibitem[CPGSV21]{Cirac_2021}
J.~Ignacio Cirac, David Pérez-García, Norbert Schuch, and Frank Verstraete.
\newblock Matrix product states and projected entangled pair states: Concepts,
  symmetries, theorems.
\newblock {\em Reviews of Modern Physics}, 93:045003, 2021.
\newblock \doi{10.1103/revmodphys.93.045003}.

\bibitem[CS23]{cudby2023gaussian}
Joshua Cudby and Sergii Strelchuk.
\newblock Gaussian decomposition of magic states for matchgate computations,
  2023, \burlalt{2307.12654}{http://arxiv.org/abs/2307.12654}.

\bibitem[DMD{\etalchar{+}}24]{denzler2023learning}
Janek Denzler, Antonio~Anna Mele, Ellen Derbyshire, Tommaso Guaita, and Jens
  Eisert.
\newblock Learning fermionic correlations by evolving with random
  translationally invariant hamiltonians.
\newblock {\em Phys. Rev. Lett.}, 133:240604, Dec 2024.
\newblock \doi{10.1103/PhysRevLett.133.240604}.

\bibitem[EA07]{Echenique_2007}
Pablo Echenique and J.~L. Alonso.
\newblock {A mathematical and computational review of Hartree–Fock SCF
  methods in quantum chemistry}.
\newblock {\em Molecular Physics}, 105(23–24):3057–3098, 2007.
\newblock \doi{10.1080/00268970701757875}.

\bibitem[EHW{\etalchar{+}}20]{Eisert_2020}
Jens Eisert, Dominik Hangleiter, Nathan Walk, Ingo Roth, Damian Markham, Rhea
  Parekh, Ulysse Chabaud, and Elham Kashefi.
\newblock Quantum certification and benchmarking.
\newblock {\em Nature Rev. Phys.}, 2:382–390, June 2020.
\newblock \doi{10.1038/s42254-020-0186-4}.

\bibitem[Eis06]{PhysRevLett.97.260501}
Jens Eisert.
\newblock Computational difficulty of global variations in the density matrix
  renormalization group.
\newblock {\em Phys. Rev. Lett.}, 97:260501, 2006.
\newblock \doi{10.1103/PhysRevLett.97.260501}.

\bibitem[FGL{\etalchar{+}}23]{fanizza2023learning}
Marco Fanizza, Niklas Galke, Josep Lumbreras, Cambyse Rouzé, and Andreas
  Winter.
\newblock Learning finitely correlated states: stability of the spectral
  reconstruction, 2023, \burlalt{2312.07516}{http://arxiv.org/abs/2312.07516}.

\bibitem[GIKL23]{grewal2023efficient}
Sabee Grewal, Vishnu Iyer, William Kretschmer, and Daniel Liang.
\newblock {Efficient learning of quantum states prepared with few non-Clifford
  gates}, 2023, \burlalt{2305.13409}{http://arxiv.org/abs/2305.13409}.

\bibitem[GKEA18]{Gluza_2018}
Marek Gluza, Martin Kliesch, Jens Eisert, and Leandro Aolita.
\newblock Fidelity witnesses for fermionic quantum simulations.
\newblock {\em Phys. Rev. Lett.}, 120:190501, May 2018.
\newblock \doi{10.1103/physrevlett.120.190501}.

\bibitem[Got98]{gottesman1998heisenberg}
Daniel Gottesman.
\newblock {The Heisenberg representation of quantum computers}, 1998,
  \burlalt{quant-ph/9807006}{http://arxiv.org/abs/quant-ph/9807006}.

\bibitem[GV08]{giuliani2008quantum}
Gabriele Giuliani and Giovanni Vignale.
\newblock {\em Quantum theory of the electron liquid}.
\newblock Cambridge University Press, 2008.
\newblock \doi{10.1017/CBO9780511619915}.

\bibitem[HG24]{hangleiter2024bell}
Dominik Hangleiter and Michael~J. Gullans.
\newblock Bell sampling from quantum circuits, 2024,
  \burlalt{2306.00083}{http://arxiv.org/abs/2306.00083}.

\bibitem[HHJ{\etalchar{+}}17]{Haah_2017}
Jeongwan Haah, Aram~W. Harrow, Zhengfeng Ji, Xiaodi Wu, and Nengkun Yu.
\newblock Sample-optimal tomography of quantum states.
\newblock {\em IEEE Trans. Inf. Th.}, page 1–1, 2017.
\newblock \doi{10.1109/tit.2017.2719044}.

\bibitem[HIN{\etalchar{+}}23]{Hinsche_2023}
Marcel Hinsche, Marios Ioannou, Alex Nietner, Jonas Haferkamp, Yihui Quek,
  Dominik Hangleiter, Jean-Pierre Seifert, Jens Eisert, and Ryan Sweke.
\newblock {One T-gate makes distribution learning hard}.
\newblock {\em Phys. Rev. Lett.}, 130:240602, 2023.
\newblock \doi{10.1103/physrevlett.130.240602}.

\bibitem[HKP20]{Huang_2020}
Hsin-Yuan Huang, Richard Kueng, and John Preskill.
\newblock Predicting many properties of a quantum system from very few
  measurements.
\newblock {\em Nature Phys.}, 16(10):1050–1057, 2020.
\newblock \doi{10.1038/s41567-020-0932-7}.

\bibitem[HLB{\etalchar{+}}24]{huang2024learning}
Hsin-Yuan Huang, Yunchao Liu, Michael Broughton, Isaac Kim, Anurag Anshu, Zeph
  Landau, and Jarrod~R. McClean.
\newblock Learning shallow quantum circuits, 2024,
  \burlalt{2401.10095}{http://arxiv.org/abs/2401.10095}.

\bibitem[JM08]{Jozsa_2008}
Richard Jozsa and Akimasa Miyake.
\newblock Matchgates and classical simulation of quantum circuits.
\newblock {\em Proc. Roy. Soc. A}, 464:3089–3106, 2008.
\newblock \doi{10.1098/rspa.2008.0189}.

\bibitem[Jon15]{RevModPhys.87.897}
Robert~O. Jones.
\newblock Density functional theory: Its origins, rise to prominence, and
  future.
\newblock {\em Rev. Mod. Phys.}, 87:897--923, 2015.
\newblock \doi{10.1103/RevModPhys.87.897}.

\bibitem[KEA{\etalchar{+}}23]{Kim2023}
Youngseok Kim, Andrew Eddins, Sajant Anand, Ken~Xuan Wei, Ewout van~den Berg,
  Sami Rosenblatt, Hasan Nayfeh, Yantao Wu, Michael Zaletel, Kristan Temme, and
  Abhinav Kandala.
\newblock Evidence for the utility of quantum computing before fault tolerance.
\newblock {\em Nature}, 618(7965):500–505, 2023.
\newblock \doi{10.1038/s41586-023-06096-3}.

\bibitem[KGE14]{PhysRevLett.113.160503}
Martin Kliesch, David Gross, and Jens Eisert.
\newblock {Matrix-product operators and states: NP-hardness and
  undecidability}.
\newblock {\em Phys. Rev. Lett.}, 113:160503, 2014.
\newblock \doi{10.1103/PhysRevLett.113.160503}.

\bibitem[Kit06]{Kitaev_2006}
Alexei Kitaev.
\newblock Anyons in an exactly solved model and beyond.
\newblock {\em Annals of Physics}, 321(1):2–111, 2006.
\newblock \doi{10.1016/j.aop.2005.10.005}.

\bibitem[Kni01]{knill2001fermionic}
Emmanuel Knill.
\newblock Fermionic linear optics and matchgates, 2001,
  \burlalt{quant-ph/0108033}{http://arxiv.org/abs/quant-ph/0108033}.

\bibitem[KR21a]{PRXQuantum.2.010201}
Martin Kliesch and Ingo Roth.
\newblock Theory of quantum system certification.
\newblock {\em PRX Quantum}, 2:010201, Jan 2021.
\newblock \doi{10.1103/PRXQuantum.2.010201}.

\bibitem[KR21b]{Kliesch_2021}
Martin Kliesch and Ingo Roth.
\newblock Theory of quantum system certification.
\newblock {\em PRX Quantum}, 2:010201, 2021.
\newblock \doi{10.1103/prxquantum.2.010201}.

\bibitem[KV94]{KearnsVALIANT}
Michael Kearns and Leslie Valiant.
\newblock Cryptographic limitations on learning boolean formulae and finite
  automata.
\newblock {\em J. ACM}, 41(1):67–95, January 1994.
\newblock \doi{10.1145/174644.174647}.

\bibitem[Lia23]{Liang_2023}
Daniel Liang.
\newblock Clifford circuits can be properly {PAC} learned if and only if
  {RP=NP}.
\newblock {\em Quantum}, 7:1036, jun 2023.
\newblock \doi{10.22331/q-2023-06-07-1036}.

\bibitem[LMMH{\etalchar{+}}17]{lanyonEfficientTomographyQuantum2017}
Ben~P. Lanyon, Christiane Maier, Tillmann~Baumgratz Milan~Holzäpfel, Cornelius
  Hempel, Petar Jurcevic, Ish Dhand, Anton~S. Buyskikh, Andew~J. Daley, Marcus
  Cramer, Martin~B. Plenio, Rainer Blatt, and Christian~F. Roos.
\newblock Efficient tomography of a quantum many-body system.
\newblock {\em Nature Phys.}, 13:1158--1162, December 2017.
\newblock \doi{10.1038/nphys4244}.

\bibitem[LOH24]{leone2023learning}
Lorenzo Leone, Salvatore F.~E. Oliviero, and Alioscia Hamma.
\newblock Learning t-doped stabilizer states.
\newblock {\em Quantum}, 8:1361, May 2024.
\newblock \doi{10.22331/q-2024-05-27-1361}.

\bibitem[LOLH24]{leone2023learning22PUBL2}
Lorenzo Leone, Salvatore F.~E. Oliviero, Seth Lloyd, and Alioscia Hamma.
\newblock Learning efficient decoders for quasichaotic quantum scramblers.
\newblock {\em Phys. Rev. A}, 109:022429, Feb 2024.
\newblock \doi{10.1103/PhysRevA.109.022429}.

\bibitem[Low22]{low2022classical}
Guang~Hao Low.
\newblock Classical shadows of fermions with particle number symmetry.
\newblock 2022, \burlalt{2208.08964}{http://arxiv.org/abs/2208.08964}.

\bibitem[LWZ{\etalchar{+}}23]{liao2023simulation}
Hai-Jun Liao, Kang Wang, Zong-Sheng Zhou, Pan Zhang, and Tao Xiang.
\newblock {Simulation of IBM's kicked Ising experiment with projected entangled
  pair operator}.
\newblock 2023, \burlalt{2308.03082}{http://arxiv.org/abs/2308.03082}.

\bibitem[Mar04]{Martin_2004}
Richard~M. Martin.
\newblock {\em Electronic structure: Basic theory and practical methods}.
\newblock Cambridge University Press, 2004.
\newblock \doi{10.1017/9781108555586}.

\bibitem[MdW18]{montanaro2018surveyquantumpropertytesting}
Ashley Montanaro and Ronald de~Wolf.
\newblock A survey of quantum property testing, 2018,
  \burlalt{1310.2035}{http://arxiv.org/abs/1310.2035}.
\newblock \urlprefix\url{https://arxiv.org/abs/1310.2035}.

\bibitem[MH25]{mele2024efficient}
Antonio~Anna Mele and Yaroslav Herasymenko.
\newblock Efficient learning of quantum states prepared with few fermionic
  non-gaussian gates.
\newblock {\em PRX Quantum}, 6(1), January 2025.
\newblock \doi{10.1103/prxquantum.6.010319}.

\bibitem[Mon17]{montanaro2017learning}
Ashley Montanaro.
\newblock {Learning stabilizer states by Bell sampling}, 2017,
  \burlalt{1707.04012}{http://arxiv.org/abs/1707.04012}.

\bibitem[NEM{\etalchar{+}}23]{Naldesi_2023}
Piero Naldesi, Andreas Elben, Anna Minguzzi, David Clément, Peter Zoller, and
  Benoît Vermersch.
\newblock Fermionic correlation functions from randomized measurements in
  programmable atomic quantum devices.
\newblock {\em Phys. Rev. Lett.}, 131(6), August 2023.
\newblock \doi{10.1103/physrevlett.131.060601}.

\bibitem[Nie23]{nietner2023free}
Alexander Nietner.
\newblock Free fermion distributions are hard to learn.
\newblock 2023, \burlalt{2306.04731}{http://arxiv.org/abs/2306.04731}.

\bibitem[O'G22]{ogorman2022fermionic}
Bryan O'Gorman.
\newblock Fermionic tomography and learning, 2022,
  \burlalt{2207.14787}{http://arxiv.org/abs/2207.14787}.
\newblock arXiv:2207.14787.

\bibitem[OLHL22]{Oliviero_2022}
Salvatore F.~E. Oliviero, Lorenzo Leone, Alioscia Hamma, and Seth Lloyd.
\newblock Measuring magic on a quantum processor.
\newblock {\em npj Quantum Information}, 8:148, 2022.
\newblock \doi{10.1038/s41534-022-00666-5}.

\bibitem[OLLH24]{leone2023learning22PUBL1}
Salvatore F.~E. Oliviero, Lorenzo Leone, Seth Lloyd, and Alioscia Hamma.
\newblock Unscrambling quantum information with clifford decoders.
\newblock {\em Phys. Rev. Lett.}, 132:080402, Feb 2024.
\newblock \doi{10.1103/PhysRevLett.132.080402}.

\bibitem[ONE13]{Efficient}
Matthias Ohliger, Vincent Nesme, and Jens Eisert.
\newblock Efficient and feasible state tomography of quantum many-body systems.
\newblock {\em New Journal of Physics}, 15:015024, 2013.
\newblock \doi{10.1088/1367-2630/15/1/015024}.

\bibitem[OW15]{odonnell2015quantum}
Ryan O'Donnell and John Wright.
\newblock Quantum spectrum testing, 2015,
  \burlalt{1501.05028}{http://arxiv.org/abs/1501.05028}.

\bibitem[PJSO23]{patra2023efficient}
Siddhartha Patra, Saeed~S. Jahromi, Sukhbinder Singh, and Roman Orus.
\newblock {Efficient tensor network simulation of IBM's largest quantum
  processors}.
\newblock 2023, \burlalt{2309.15642}{http://arxiv.org/abs/2309.15642}.

\bibitem[PV88]{PittValiant1988}
Leonard Pitt and Leslie~G. Valiant.
\newblock Computational limitations on learning from examples.
\newblock {\em J. ACM}, 35(4):965–984, October 1988.
\newblock \doi{10.1145/48014.63140}.

\bibitem[RAS{\etalchar{+}}19]{rocchetto_2019_PAC_experiment}
Andrea Rocchetto, Scott Aaronson, Simone Severini, Gonzalo Carvacho, Davide
  Poderini, Iris Agresti, Marco Bentivegna, and Fabio Sciarrino.
\newblock Experimental learning of quantum states.
\newblock {\em Science Advances}, 5:eaau1946, 2019.
\newblock \doi{10.1126/sciadv.aau1946}.

\bibitem[RF23]{rouzé2023learning}
Cambyse Rouzé and Daniel~Stilck França.
\newblock Learning quantum many-body systems from a few copies, 2023,
  \burlalt{2107.03333}{http://arxiv.org/abs/2107.03333}.

\bibitem[RFHC23]{rudolph2023classical}
Manuel~S. Rudolph, Enrico Fontana, Zoë Holmes, and Lukasz Cincio.
\newblock Classical surrogate simulation of quantum systems with lowesa.
\newblock 2023, \burlalt{2308.09109}{http://arxiv.org/abs/2308.09109}.

\bibitem[RHBM13]{Rossi_2013}
Massimiliano Rossi, Marcus Huber, Dagmar Bruß, and Chiara Macchiavello.
\newblock Quantum hypergraph states.
\newblock {\em New Journal of Physics}, 15:113022, November 2013.
\newblock \doi{10.1088/1367-2630/15/11/113022}.

\bibitem[Roc18]{rocchetto2018stabiliser}
Andrea Rocchetto.
\newblock {Stabiliser states are efficiently PAC-learnable}.
\newblock 2018, \burlalt{1705.00345}{http://arxiv.org/abs/1705.00345}.

\bibitem[Sch18]{schrieffer2018theory}
J.~Robert Schrieffer.
\newblock {\em Theory of superconductivity}.
\newblock CRC press, 2018.
\newblock \doi{10.1201/9780429495700}.

\bibitem[SCV08]{PhysRevLett.100.250501}
Norbert Schuch, Ignacio Cirac, and Frank Verstraete.
\newblock Computational difficulty of finding matrix product ground states.
\newblock {\em Phys. Rev. Lett.}, 100:250501, 2008.
\newblock \doi{10.1103/PhysRevLett.100.250501}.

\bibitem[ST22]{Surace_2022}
Jacopo Surace and Luca Tagliacozzo.
\newblock {Fermionic Gaussian states: an introduction to numerical approaches}.
\newblock {\em {SciPost} Phys. Lect. Notes}, page~54, may 2022.
\newblock \doi{10.21468/scipostphyslectnotes.54}.

\bibitem[SV09]{HFNP}
Norbert Schuch and Frank Verstraete.
\newblock Computational complexity of interacting electrons and fundamental
  limitations of density functional theory.
\newblock {\em Nature Phys.}, 5:732, 2009.
\newblock \doi{10.1038/nphys1370}.

\bibitem[SWVC07]{PhysRevLett.98.140506}
Norbert Schuch, Michael~M. Wolf, Frank Verstraete, and J.~Ignacio Cirac.
\newblock Computational complexity of projected entangled pair states.
\newblock {\em Phys. Rev. Lett.}, 98:140506, 2007.
\newblock \doi{10.1103/PhysRevLett.98.140506}.

\bibitem[TD02]{Terhal_2002}
Barbara~M. Terhal and David~P. DiVincenzo.
\newblock Classical simulation of noninteracting-fermion quantum circuits.
\newblock {\em Phys. Rev. A}, 65:032325, March 2002.
\newblock \doi{10.1103/physreva.65.032325}.

\bibitem[Tov84]{TOVEY198485}
Craig~A. Tovey.
\newblock {A simplified NP-complete satisfiability problem}.
\newblock {\em Discrete Applied Mathematics}, 8:85--89, 1984.
\newblock \doi{https://doi.org/10.1016/0166-218X(84)90081-7}.

\bibitem[Val01]{Valiant}
{\em STOC '01: Proceedings of the Thirty-Third Annual ACM Symposium on Theory
  of Computing}, New York, NY, USA, 2001. Association for Computing Machinery.

\bibitem[WHLB23]{wan2023matchgate}
Kianna Wan, William~J. Huggins, Joonho Lee, and Ryan Babbush.
\newblock Matchgate shadows for fermionic quantum simulation.
\newblock 2023, \burlalt{2207.13723}{http://arxiv.org/abs/2207.13723}.

\bibitem[WJEH21]{Windt_2021}
Bennet Windt, Alexander Jahn, Jens Eisert, and Lucas Hackl.
\newblock {Local optimization on pure Gaussian state manifolds}.
\newblock {\em {SciPost} Phys.}, 10:066, mar 2021.
\newblock \doi{10.21468/scipostphys.10.3.066}.

\bibitem[WZ14]{Whitfield_2014}
James~Daniel Whitfield and Zoltán Zimborás.
\newblock {On the NP-completeness of the Hartree-Fock method for
  translationally invariant systems}.
\newblock {\em The Journal of Chemical Physics}, 141:234103, 2014.
\newblock \doi{10.1063/1.4903453}.

\bibitem[Yog19]{yoganathan2019condition}
Mithuna Yoganathan.
\newblock A condition under which classical simulability implies efficient
  state learnability.
\newblock 2019, \burlalt{1907.08163}{http://arxiv.org/abs/1907.08163}.

\bibitem[ZRM21]{Zhao_2021}
Andrew Zhao, Nicholas~C. Rubin, and Akimasa Miyake.
\newblock Fermionic partial tomography via classical shadows.
\newblock {\em Phys. Rev. Lett.}, 127(11):110504, 2021.
\newblock \doi{10.1103/physrevlett.127.110504}.

\end{thebibliography}
\let\oldaddcontentsline\addcontentsline
\renewcommand{\addcontentsline}[3]{}
\medskip

\end{document}